\documentclass[11pt, a4paper,UKenglish]{article}

\usepackage{booktabs} % For formal tables
\usepackage[noend,ruled]{algorithm2e} % For algorithms

\SetAlFnt{\small}
\SetAlCapFnt{\small}
\SetAlCapNameFnt{\small}
\SetAlCapHSkip{0pt}
\IncMargin{-\parindent}

\usepackage{hyperref}
\usepackage[utf8]{inputenc}
\usepackage[margin=1in]{geometry}

\usepackage{dsfont}
\usepackage{amsmath,amsthm, amsfonts, amssymb}
\usepackage{mathtools}
\usepackage{thmtools, thm-restate}
\newtheorem{lemma}{Lemma}

\newtheorem{corollary}{Corollary}
\newtheorem{proposition}{Proposition}
\newtheorem{observation}{Observation}

\theoremstyle{definition}

\usepackage{natbib}
\usepackage{xcolor}

\newcommand{\OPT}{\mathrm{SW}_{\mathrm{opt}}}
\newcommand{\ALG}{\mathrm{SW}_{\mathrm{pp}}}
\newcommand{\revpp}{\mathrm{revenue}_{\mathrm{pp}}}
\newcommand{\revopt}{\mathrm{revenue}_{\mathrm{opt}}}
\newcommand{\Ex}[2][]{\mbox{\rm\bf E}_{#1}\left[#2\right]}
\renewcommand{\Pr}[2][]{\mbox{\rm\bf Pr}_{#1}\left[#2\right]}
\newcommand{\e}{\mathrm{e}}
\newcommand{\growingmid}{\mathrel{}\middle|\mathrel{}}
\newcommand{\transformedquantile}{R}
\newcommand{\quant}{z}

\title{Asymptotically Optimal Welfare of Posted Pricing for Multiple Items with MHR Distributions}
\author{Alexander Braun \footnote{alexander.braun@uni-bonn.de} \qquad Matthias Buttkus \qquad Thomas Kesselheim \footnote{thomas.kesselheim@uni-bonn.de} \\{\small Institute of Computer Science, University of Bonn}}

\begin{document}
	\maketitle
	\begin{abstract}
		We consider the problem of posting prices for unit-demand buyers if all $n$ buyers have identically distributed valuations drawn from a distribution with monotone hazard rate. We show that even with multiple items asymptotically optimal welfare can be guaranteed.

Our main results apply to the case that either a buyer's value for different items are independent or that they are perfectly correlated. We give mechanisms using dynamic prices that obtain a $1 - \Theta \left( \frac{1}{\log n}\right)$-fraction of the optimal social welfare in expectation. Furthermore, we devise mechanisms that only use static item prices and are $1 - \Theta \left( \frac{\log\log\log n}{\log n}\right)$-competitive compared to the optimal social welfare. As we show, both guarantees are asymptotically optimal, even for a single item and exponential distributions.
	\end{abstract}
	\newpage
	% !TEX root = mhr-prophet.tex
\section{Introduction}
\label{section:introduction}

Posting prices is a very simple way to de-centralize markets. One assumes that buyers arrive sequentially. Whenever one of them arrives, a mechanism offers a menu of items at suitably defined prices. The buyer then decides to accept any offer, depending on what maximizes her own utility. Such a mechanism is incentive compatible by design, usually easy to explain and can be implemented online. For this reason, there is a large interest in understanding what social welfare and revenue can be guaranteed in comparison to mechanisms that optimize the respective objective.

Let us consider the following setting: There is a set of $m$ heterogeneous items $M$, each of which we would like to allocated to one of $n$ buyers. Each buyer $i$ has a private valuation function $v_i\colon 2^M \to \mathbb{R}_{\geq 0}$. We assume that valuation functions are unit-demand. That is, $v_i(S) = \max_{j \in S} v_i(\{ j \})$, meaning that the value a buyer associates to a set is simply the one of the most valuable item in this set. Let $\OPT = \max_{\text{allocations $(S_1, \ldots, S_n)$}} \sum_{i = 1}^n v_i(S_i)$ denote the optimal (offline/ex-post) social welfare. Note that this optimal solution is nothing but the maximum-weight matching in a bipartite graph in which all buyers and items correspond to a vertex each and an edge between the vertices of buyer $i$ and item $j$ has weight $v_i(\{j\})$.

To capture the pricing setting, we assume that the functions $v_1, \ldots, v_n$ are unknown a priori; all of them are drawn independently from the same, publicly known distribution. For every item, one can either set a static item price or change the prices dynamically over time. Buyers arrive one-by-one and each of them chooses the set of items that maximizes her utility given the current prices among the remaining items. Static prices have the advantage that they are easier to explain and thus give easier mechanisms. However, dynamic prices can yield both higher welfare and revenue because they can be adapted to the remaining supply and the remaining number of buyers to appear.

Coming back to the interpretation of a bipartite matching problem, a posted-prices mechanism corresponds to an online algorithm, where the buyers correspond to online vertices and the items correspond to offline vertices. However, not every online algorithm necessarily corresponds to a posted-prices mechanism: There might not be item prices such that the choices of the algorithm correspond to the ones by a buyer maximizing their utility.

We would like to understand which fraction of the optimal (offline/ex-post) welfare posted-prices mechanisms can guarantee. The case of a single item is well understood via \emph{prophet inequalities} from optimal stopping theory. Let us call a posted-prices mechanism $\beta$-competitive (with respect to social welfare) if its expected welfare $\Ex{\ALG}$ is at least $\beta \Ex{\OPT}$. For a static price and a single item, the best such guarantee is $\beta = 1 - \frac{1}{e} \approx 0.63$ \citep{DBLP:conf/sigecom/CorreaFHOV17,DBLP:conf/soda/EhsaniHKS18}; for dynamic pricing and a single item, it is $\beta \approx 0.745$ \citep{DBLP:conf/stoc/AbolhassaniEEHK17,DBLP:conf/sigecom/CorreaFHOV17}. There are a number of extensions of these results to multiple items (see Section~\ref{sec:related_work} for details), also going beyond unit-demand valuations, many of which are $O(1)$-competitive.

The competitive ratios of $\beta = 1 - \frac{1}{e} \approx 0.63$ and $\beta \approx 0.745$ are optimal in the sense that there are distributions and choices of $n$ such that no better guarantee can be obtained. Importantly, they are still tight when imposing a lower bound on $n$. That is, even for large $n$, there is a distribution such that if all values are drawn from this distribution the respective bound cannot be beaten.

\subsection{Distributions with Monotone Hazard Rate}

In this paper, we strengthen previous results by restricting the class of distributions to ones with monotone hazard rate. The single-item case is defined as follows. Consider a probability distribution on the reals with probability density function (PDF) $f$ and cumulative distribution function (CDF) $F$, its hazard rate $h$ is defined by $h(x) = \frac{f(x)}{1 - F(x)}$ for $x$ with $F(x) < 1$. It has a \emph{monotone hazard rate} (MHR) --- more precisely, increasing hazard rate --- if $h$ is a non-decreasing function. It has become a common and well-studied approach to model buyer preferences by MHR distributions. One of the reasons is that many standard distributions exhibit a monotone hazard rate such as, for example, uniform, normal, exponential and logistic distributions. (For a much more extensive list see \citet{rinne2014hazard}.) Furthermore, the monotone hazard rate of distributions is also preserved under certain operations; for example, order statistics of MHR distributions also have an MHR distribution. Additionally, every MHR distribution is regular in the sense that its virtual value function \citep{DBLP:conf/sigecom/ChawlaHK07, DBLP:journals/mor/Myerson81} is increasing. 

We generalize results to multiple items and consider two fundamental settings. On the one hand, we consider \emph{independent item valuations}, i.e. $v_{i,j} \sim \mathcal{D}_j$ is an independent draw from a distribution $\mathcal{D}_j$. In other words, the value of item $j$ is independent of the value of item $j'$ and both values are drawn from (possibly different) MHR distribution as defined above. On the other hand, we assume correlated values for items via the notion of \emph{separable item valuations}, which are common in ad auctions~\citep{EdelmanOS07,Varian07}: Each buyer has a type $t_i \geq 0$ and each item has an item-dependent multiplier $\alpha_j$ where now $v_{i,j} = \alpha_j \cdot t_i$. Again, $t_i \sim \mathcal{D}$ and $\mathcal{D}$ is a distribution with monotone hazard rate. We note that this case subsumes and extends the case of $k$ identical items.

As we show, in these cases, asymptotically optimal welfare can be guaranteed. That is, if $n$ grows large, the social welfare when suitably choosing prices is within a $1 - o(1)$ factor of the optimum, where the $o(1)$ term is independent of the distribution as long as its marginals satisfy the MHR property. Stated differently, there is a sequence $(\beta_n)_{n \in \mathbb{N}}$ with $\beta_n \to 1$ for $n \to \infty$ such that for every number of buyers $n$ there exists a posted-prices mechanism that takes any distribution with MHR as input and guarantees $\Ex{\ALG} \geq \beta_n \Ex{\OPT}$. As pointed out before, such a result does not hold for arbitrary, non-MHR distributions. Even with a single item, the limit is then upper-bounded by $\approx 0.745$.

A similar effect has already been observed by \citet{DBLP:conf/wine/GiannakopoulosZ18}. They show that the revenue of static pricing for a single item with MHR distributions asymptotically reaches the optimal revenue. In contrast, our results concern welfare. Still, some of our results also have implications for revenue, either because we bound the revenue or because one could apply the results to virtual values.

\subsection{Our Results and Techniques}
We design mechanisms for both independent and separable item valuations. The ones using dynamic prices ensure a $\left( 1 - O\left(\frac{1}{\log n}\right) \right)$-fraction of the expected optimal social welfare. The ones using static prices guarantee a $\left( 1 - O\left(\frac{ \log  \log  \log n}{ \log n} \right) \right)$-fraction.  We also show that these guarantees are best possible, even in the case of only a single item. Note that the bounds are independent of the number of items $m$, which may also grow in the number of buyers $n$.

\noindent\textbf{Independent Valuations (Section~\ref{section:independent_valuations}).} The technically most interesting result is the one on dynamic pricing when values are independent across items. The idea is to set prices so that the offline optimum is mimicked. If item $j$ is allocated in the optimal allocation with probability $q_j$, then we would like it to be sold in every step with an ex-ante probability of $\frac{q_j}{n}$. However, analyzing such a selling process is still difficult because items are incomparable and bounds for MHR distributions cannot be applied directly to draws from multiple distributions, which are not necessarily identical. To bypass this problem, we introduce a reduction that allows us to view item valuations not only as independent but also as identically distributed. To this end, we compare the selling process of our mechanism to a hypothetical setting, in which buyers do not make their decisions based on the actual utility but in quantile space. We observe that the revenue of both is identically distributed and utility is maximized in the former mechanism. As a consequence, the welfare obtained by the quantile allocation rule is a feasible lower bound on the welfare of the sequential posted-prices mechanism. Only afterwards, we can apply a concentration bound due to the MHR restriction. 

The idea of our mechanism using static prices is to set prices suitably high in order to bound the revenue of our mechanism with a sufficient fraction of the optimum. While all other bounds apply for any number of items $m$, this bound unfortunately requires $m \leq \frac{n}{\left(\log \log n\right)^2}$. We leave it as an open problem to extend the result for larger number of items.

\noindent\textbf{Separable Valuations (Section~\ref{section_alpha_correlation}).} 
Our way of setting dynamic prices in the case of separable valuations is similar to the approach in independent valuations. This setting is even a little simpler because we can assume without loss of generality that there are as many items as buyers. Our pricing strategy ensures that in each step each item is sold equally likely as well as one item is sold for sure. In the analysis, we observe that we match a buyer and an item if the quantile of the buyer's value is in a specific range. Now, the MHR property comes into play which allows to bound quantiles of the distribution in a suitable way.

In the static case, to lower-bound the welfare of our mechanism, we compare it to the one of the VCG mechanism \citep{RePEc:kap:pubcho:v:11:y:1971:i:1:p:17-33, RePEc:ecm:emetrp:v:41:y:1973:i:4:p:617-31, RePEc:bla:jfinan:v:16:y:1961:i:1:p:8-37} which maximizes social welfare. To this end, we split social welfare in revenue and utility and bound each quantity separately. That is, we relate the revenue and the sum of buyer utilities of our posted-pricing mechanism to the ones of the VCG mechanism separately. For the revenue, we set prices fairly low to ensure that we sell all items with reasonably high probability. Still, these prices are high enough to use the MHR property and derive a suitable lower bound of the prices. The utility comparison is more complicated, we solve this issue by an unusual application of the equality of expected revenue and virtual welfare due to \citet{DBLP:journals/mor/Myerson81}.

\noindent\textbf{Optimality (Section~\ref{section:optimality}).} The achieved bounds on the competitive ratio are optimal for both dynamic as well as static pricing. We show this by considering the single-item case with an exponential distribution, which is a special case of both independent and separable valuations. For dynamic prices, we use the correspondence to a Markov decision process showing that no online algorithm is better than $1- \Omega \left( \frac{1}{\log n} \right)$-competitive. Then we also show that the competitive ratio cannot be better than $1 - \Omega \left( \frac{\log \log \log n}{\log n} \right)$ for any choice of a static price by writing out the expected social welfare explicitly.

\noindent\textbf{Subadditive Valuations (Section~\ref{section:subadditive}).} We also demonstrate that our techniques are applicable beyond unit-demand settings by giving mechanisms for the more general class of subadditive valuation functions. Our dynamic pricing mechanism is $1 - O \left( \frac{1+\log m}{\log n}\right)$-competitive for subadditive buyers. We complement this by a static pricing mechanism which is $1 - O \left( \frac{\log\log\log n}{\log n} + \frac{\log m}{\log n}\right)$-competitive.
Both guarantees can be derived by showing that the revenue of the posted pricing mechanism is at least as high as the respective fraction of the optimal social welfare. As a consequence, these bounds directly imply the competitive ratios for welfare and revenue. For small $m$, these bounds are again tight by our optimality results. Obtaining tight bounds for large $m$ still remains an open problem.

\subsection{Further Related Work}
\label{sec:related_work}
% !TEX root = mhr-prophet.tex

As mentioned already, our setup restricted to a single item is highly related to prophet inequalities. Prophet inequalities have their origin in optimal stopping theory, dating back to the 1970s \citep{krengel1977semiamarts}. Only much later they were considered as a tool to understand the loss by posting prices as opposed to using other mechanisms. In this context, Samuel-Cahn's result~\citep{Samuel84} then got the interesting interpretation that posting an appropriately chosen static price for an item is $\frac{1}{2}$-competitive for any buyer distributions; different buyers may even be drawn from different distributions. This guarantee is optimal, even for dynamic pricing.

Improvements for the single-item case are only possible by imposing further assumptions. Most importantly, this concerns the case that all buyer values are drawn from the same distribution. While already discussed by \citet{HillKertz82}, this problem has been solved only very recently by devising an $\approx 0.745$-competitive mechanism that relies on dynamic pricing \citep{DBLP:conf/stoc/AbolhassaniEEHK17,DBLP:conf/sigecom/CorreaFHOV17}. By using static pricing, one cannot be better than $1 - \frac{1}{e} \approx 0.63$-competitive \citep{DBLP:conf/sigecom/CorreaFHOV17,DBLP:conf/soda/EhsaniHKS18}.

Better guarantees can also be achieved by assuming that multiple, identical items are for sale. In this case, one can use concentration results. The respective competitive ratios tend to $1$ for a growing number of item copies. \citet{DBLP:conf/aaai/HajiaghayiKS07} gave the first guarantee for such a setting, \citet{DBLP:journals/siamcomp/Alaei14} later improved it to tightness.

For identical regular distributions, all of the above results also apply to welfare as well as revenue maximization because the prices can be imposed in the space of the virtual valuation \citep{DBLP:journals/mor/Myerson81}; also see \citet{DBLP:conf/stoc/ChawlaHMS10}. Also impossibility results transfer \citep{DBLP:journals/orl/CorreaFPV19}.

When it comes to multiple, heterogeneous items, there is a significant difference between welfare and revenue maximization because Myerson's characterization does not apply anymore. For welfare maximization, \citet{DBLP:conf/soda/FeldmanGL15} show that static item prices still yield a competitive ratio of $\frac{1}{2}$ even for XOS valuations and not necessarily identically distributed buyer valuations. Concerning subadditive valuation functions, \citet{zhang:LIPIcs:2020:12948} give a $O(\log m / \log \log m)$-competitive prophet inequality, \citet{10.1145/3440968.3440972} show how to obtain a competitive ratio of $O(\log \log m)$. The only improvement for identically distributed buyers is to $1 - \frac{1}{e}$ for unit-demand buyers based on dynamic pricing \citep{DBLP:conf/soda/EhsaniHKS18}. Among others, \citet{DBLP:conf/stoc/ChawlaDHKMS17} considered a combinatorial generalization of such a setting with many item copies (see Lucier's survey \citep{DBLP:journals/sigecom/Lucier17} on a broader overview of combinatorial generalizations).

For revenue maximization, one usually imposes the additional assumption that items are independent. This makes it possible to also apply prophet inequalities on the sequence of items rather than buyers and thus maximize revenue for unit-demand buyers via posted prices \citep{DBLP:conf/sigecom/ChawlaHK07,DBLP:conf/stoc/ChawlaHMS10}. \citet{DBLP:conf/stoc/CaiZ17} consider more general XOS and subadditive valuations and apply a duality framework instead. They design a posted-prices mechanism with an entry fee that gives an $O(1)$ or $O(\log m)$ approximation to the optimal revenue. In \citet{10.1145/3440968.3440972}, the approximation of the optimal revenue for subadditive valuations is improved to $O(\log \log m)$.

There are surprisingly few results on pricing and prophet inequalities that derive better guarantees by imposing additional assumptions on the distribution. \citet{Babaioff:2017:ACM:3037382} consider the problem of maximizing revenue when selling a single item to one of $n$ buyers drawn i.i.d.\ from an \emph{unknown} MHR distribution with a bounded support $[1, h]$. If $n$ is large enough compared to $h$, they get a constant-factor approximation to the optimal revenue using dynamic posted prices. Note that in contrast, in our paper, we assume to know the underlying distributions perfectly. \citet{DBLP:conf/wine/GiannakopoulosZ18} consider revenue maximization in the single-item setting with valuations drawn independently from the same MHR distribution. They show that by offering the item for the same static price to all bidders one can achieve asymptotically optimal revenue. More precisely, one of their main results is that one gets within a factor of $1 - O\left(\frac{\ln \ln n}{\ln n}\right)$. While they claim this result is ``essentially tight'', we show that the best factor is indeed $1 - \Theta\left(\frac{\ln \ln \ln n}{\ln n}\right)$ because it is a special case of our results (see Section~\ref{section:subadditive}). It is not clear, how one could apply their result to welfare maximization as the MHR property is not preserved when moving between virtual and actual values. Furthermore, their results do not admit any apparent generalization to multiple items. \citet{10.1007/978-3-030-35389-6_17} also consider revenue maximization in the single-item setting with identical and independent MHR values but in a non-asymptotic sense, providing a bound for every $n$.

	% !TEX root = mhr-prophet.tex
\section{Preliminaries}
\label{section:preliminiaries}
We consider a setting of $n$ buyers $N$ and a set $M$ of $m$ items. Every buyer has a valuation function $v_i\colon 2^M \to \mathbb{R}_{\geq 0}$ mapping each bundle of items to the buyer's valuation. We assume buyers to be unit-demand, that is $v_i(S) = \max_{j \in S} v_{i,j}$. The functions $v_1, \ldots, v_n$ are unknown a priori but all drawn independently from the same, publicly known distribution $\mathcal{D}$. Let $\mathcal{D}_j$ be the marginal distribution of $v_{i,j}$, which is the value of a buyer for being allocated item $j$. We assume that $\mathcal{D}_j$ is a continuous, real, non-negative distribution with monotone hazard rate. That is, let $F_j$ be the cumulative distribution function of $\mathcal{D}_j$ and $f_j$ its probability density function. The distribution's hazard rate is defined as $h_j(x) = f_j(x)/(1-F_j(x))$ for all $x$ such that $F_j(x) < 1$. We assume a \emph{monotone hazard rate}, which means that $h_j$ is a non-decreasing function. Equivalently, we can require $x \mapsto \log(1 - F_j(x))$ to be a concave function.

We design \emph{posted-prices mechanisms}. That is, the buyers arrive one by one in order $1, \ldots, n$. In the $i$-th step, buyer $i$ arrives and has the choice between all items which have not been allocated so far. Let $M^{(i)}$ denote this set of available items. The mechanism presents the $i$-th buyer a menu of prices $p^{(i)}_j$ for all items $j \in M^{(i)}$. The buyer then picks the item $j_i \in M^{(i)}$ which maximizes her \emph{utility} $v_{i,j_i} - p_{j_i}^{(i)}$ if positive\footnote{We can assume that any buyer is buying at most one item as buyers are unit-demand. Hence, no buyer can increase utility by buying a second (lower valued) item.}. Buyer $i$ and item $j_i$ are matched immediately and irrevocably. If buyer $i$ has negative utility for all items $j \in M^{(i)}$, then buyer $i$ does not buy any item and remains unmatched. Generally, the prices for buyer $i$ may depend arbitrarily on $M^{(i)}$ and the distribution $\mathcal{D}$. We call prices \emph{static} if there are $p_1, \ldots, p_m$ such that $p^{(i)}_j = p_j$ for all $i$ and all $j$. 

Fix any posted-prices mechanism and let $j_i$ denote the item allocated to buyer $i$ (set $j_i = \perp$ if $i$ remains unmatched in the mechanism). The \emph{expected social welfare} of the mechanism is given by $\Ex[]{\sum_{i = 1}^n v_{i,j_i}  } =: \Ex[]{\ALG}$.
In comparison, let the social welfare maximizing allocation assign item $j_i^\ast$ to buyer $i$. Its expected social welfare is therefore given by $\Ex[]{\sum_{i = 1}^n v_{i,j_i^\ast} } =: \Ex[]{\OPT}$. 

We call a posted-prices mechanism \emph{$\beta$-competitive} if it ensures that the expected social welfare of its allocation is at least a $\beta$-fraction of the expected optimal social welfare. That is, for any choice of distribution,
\[
\Ex{\ALG} =\Ex[]{\sum_{i = 1}^n v_{i,j_i}  } \geq \beta \Ex[]{\sum_{i = 1}^n v_{i,j_i^\ast}} = \beta \Ex{\OPT} \enspace.
\]

	% !TEX root = mhr-prophet.tex
\section{Asymptotically Tight Bounds for Independent Valuations}
\label{section:independent_valuations}
In this section, we show how to derive bounds if the buyers' values are independent across items. That is, each $v_{i,j} \sim \mathcal{D}_j$ is drawn independently from a distribution with monotone hazard rate. This is a standard assumption when considering multiple items \citep{DBLP:conf/sigecom/ChawlaHK07, DBLP:conf/stoc/ChawlaHMS10}. As a consequence, the distribution over valuations is a product distribution $v_i = \left( v_{i,1},\dots,v_{i,m} \right) \sim \mathcal{D} = \prod_{j=1}^{n} \mathcal{D}_j$ for any $i \in N$ and every $\mathcal{D}_j$ satisfies the MHR condition. 

\subsection{Dynamic prices}
\label{section:independent_valuations_dynamic}
We first consider the case of dynamic pricing mechanisms. Without loss of generality, we can assume that $m \geq n$. If we have less items than buyers, i.e. $m < n$, we can add dummy items with value $0$ to ensure $m=n$. Matching $i$ to one of these dummy items in the mechanism then corresponds to leaving $i$ unmatched. Observe that technically a point mass on $0$ is not a MHR distribution. However, all relevant statements still apply.

Our mechanism is based on a pricing rule which balances the probability of selling a specific item. To this end, let $ M^{(i)}$ be the set of remaining items as buyer $i$ arrives. We determine dynamic prices such that one item is allocated for sure in every step. Therefore, always $\lvert M^{(i)} \rvert = m - i + 1$. We can now define $q_{j}^{(i)}$ to be the probability that item $j$ is allocated in the ``remaining'' offline optimum on $M^{(i)}$ and $n-i+1$ buyers if $j \in M^{(i)}$ and $0$ else. In other words, if $j \in M^{(i)}$, $q_{j}^{(i)}$ is the probability that item $j$ is allocated in the offline optimum constrained to buyers $1, \ldots, i-1$ receiving the items from $M \setminus M^{(i)}$. The prices $(p^{(i)}_j)_{j \in M^{(i)}}$ are now chosen such that buyer $i$ buys item $j$ with probability $\frac{q^{(i)}_j}{n-i+1}$ and one item is allocated for sure. To see that such prices exist, observe the following: fix any price vector $\mathbf{x} = \left( x_j \right)_{j \in M^{(i)}}$ and denote by $r^{(i)}_j(\mathbf{x}) = \Pr[]{i \text{ buys item } j \text{ at prices } \mathbf{x} \growingmid M^{(i)} }$. As the random variables $v_{i,j}$ are continuous and independent, the probability that buyer $i$ buys item $j$ at prices $\mathbf{x}$ given the current set of items $M^{(i)}$ is continuous in $x_j$. Hence, we can consider the mapping $\left(\phi^{(i)}( \mathbf{x} ) \right)_j = \frac{  n-i+1 }{q^{(i)}_j} \cdot r^{(i)}_j(\mathbf{x}) \cdot x_j$ for any $j \in M^{(i)}$ which is also continuous and hence, by the use of Brouwer's fixed point theorem\footnote{In addition, we can use that prices $\mathbf{x} = \left( x_j \right)_{j \in M^{(i)}}$ are always bounded by $0 \leq x_j \leq F_j^{-1}\left( 1 - \frac{q^{(i)}_j}{n-i+1} \right)$ to get a convex and compact set of price vectors.} has our desired price vector $(p^{(i)}_j)_{j \in M^{(i)}}$ as fixed point. This allows us to state the following theorem.

\begin{restatable}{theorem}{theoremindependentdynamic}
	\label{theorem_independent_dynamic}
	The posted-prices mechanism with dynamic prices and independent item-valuations is $1 - O \left( \frac{1}{\log n} \right)$-competitive with respect to social welfare.
\end{restatable}

Note that in the case $m \leq n$ we will always have $q^{(i)}_j = 1$ for $j \in M^{(i)}$, which significantly simplifies the argument. The proof for the general case can be found in Appendix~\ref{appendix:independent_valuations_dynamic_general}. Here, we give a sketch with the major steps and key techniques. \\ 

In order to bound the social welfare obtained by the posted-prices mechanism, we consider the following \emph{quantile allocation rule}. For any $j \in M^{(i)}$ with $q_{j}^{(i)} > 0$, compute $R_{j}^{(i)} := F_j(v_{i, j})^{\frac{1}{q_j^{(i)}}}$ and allocate buyer $i$ the item $j$ which maximizes $R_{j}^{(i)}$. Observe that by this definition for any $i$, any $j$ and any $t \in [0,1]$,
\begin{align*}
	\Pr[]{R_{j}^{(i)} \leq t} = \Pr[]{F_j(v_{i,j}) \leq t^{q_j^{(i)}}} = \Pr[]{v_{i,j} \leq F_j^{-1}\left( t^{q_j^{(i)}} \right)} = F_j \left( F_j^{-1} \left( t^{q_j^{(i)}} \right) \right) = t^{q_j^{(i)}} \enspace.
\end{align*}
Note that for $q_j^{(i)} = 1$, this is exactly the CDF of a random variable drawn from $\text{Unif}[0,1]$. Define indicator variables $X_{i,j}$ which are $1$ if buyer $i$ is allocated item $j$ in the quantile allocation rule. Then, we can observe the following.

\begin{restatable}{observation}{observationallocationprobabilities}
	\label{observation:allocation_probabilities}
	It holds \[ \Pr[]{ X_{i,j} = 1 \growingmid M^{(i)} } =  \frac{q^{(i)}_j}{n-i+1} \enspace. \]
\end{restatable}

Note that by this, the probability of allocating item $j$ in step $i$ via the quantile allocation rule is $\frac{q^{(i)}_j}{n-i+1}$, exactly as in the posted-prices mechanism.

\begin{proof}
	We allocate item $j$ in the quantile allocation rule if $R_{j}^{(i)} \geq R_{j'}^{(i)}$ for any $j' \in M^{(i)}$. For fixed $M^{(i)}$, also the values of $q_j^{(i)}$ are fixed. Hence, we can use independence of the $v_{i,j}$ variables to compute:
	\begin{align*}
	\Pr[]{ X_{i,j} = 1 \growingmid M^{(i)} } &= \Pr[]{ \max_{ j' \neq j }R_{j'}^{(i)} \leq R_{j}^{(i)} \growingmid M^{(i)} } = \int_{0}^{1} \Pr[]{\max_{ j' \neq j }R_{j'}^{(i)} \leq t \growingmid M^{(i)} } q_j^{(i)} t^{q_j^{(i)} -1} dt \\ & = \int_{0}^{1} \prod_{j' \neq j} \Pr[]{R_{j'}^{(i)} \leq t \growingmid M^{(i)} } q_j^{(i)} t^{q_j^{(i)} -1} dt = \int_{0}^{1} \left( \prod_{j' \neq j} t^{q_{j'}^{(i)}} \right) q_j^{(i)} t^{q_j^{(i)} -1} dt \\& = q_j^{(i)} \int_{0}^{1} t^{(n-i+1)-1} dt = \frac{q^{(i)}_j}{n-i+1} \enspace,
	\end{align*}
	where we use that $\sum_{j \in M^{(i)}} q_j^{(i)} = n-i+1$ for any value of $i$.	
\end{proof}

Now, the crucial observation is that the expected contribution of any buyer to the social welfare in the posted-prices mechanism is at least as large as under the quantile allocation rule. To see this, fix buyer $i$ and split buyer $i$'s contribution to the social welfare into revenue and utility. Concerning revenue, note that in both cases the probability of selling any item $j$ to buyer $i$ is equal to $\frac{q^{(i)}_j}{n-i+1}$ and we allocate one item for sure. So, the expected revenue is identical. Further, since we maximize utility in the posted-prices mechanism, the achieved utility is always at least the utility of the quantile allocation rule. So, overall, we get $\Ex[]{\textnormal{SW}_{\textnormal{quantile}}} \leq \Ex{\ALG}$. 

Next, we aim to control the distribution of $v_{i,j}$ given that $X_{i,j} = 1$ in order to get access to the value of an agent being allocated an item in the quantile allocation rule. To this end, we use the following lemma.

\begin{restatable}{lemma}{lemmaquantileallocation}
	\label{lemma_quantile_allocation}
	For all $i$, $j$ and $M^{(i)}$, we have
	\[
	\Pr{v_{i,j} \leq t \growingmid X_{i,j} = 1, M^{(i)}} = F_j(t)^{\frac{n-i+1}{q_j^{(i)}}} \enspace.
	\]
\end{restatable}
\label{proof_lemma_quantile_allocation}
\begin{proof}
	Observe that in the vector $\left( \transformedquantile_{j}^{(i)} \right)_{j \in M^{(i)}}$, we choose $j$ to maximize $\transformedquantile_{j}^{(i)}$. Now, for any value $v_{i,j'}$, we consider the following transform $\psi_j$: For any $j' \in M^{(i)}$, define \[ \psi_j(v_{i,j'}) := F_j^{-1} \left( \left( \transformedquantile_{j'}^{(i)} \right)^{q_j^{(i)}} \right) \enspace. \] Observe that for $j' = j$, we get that \[ \psi_j(v_{i,j}) = F_j^{-1} \left( \left( \transformedquantile_{j}^{(i)} \right)^{q_j^{(i)}} \right) = F_j^{-1} \left( \left(  F_j(v_{i, j})^{\frac{1}{q_j^{(i)}}}  \right)^{q_j^{(i)}} \right) =  F_j^{-1} \left(  F_j(v_{i, j})  \right) = v_{i,j} \enspace.\] Further, we can compute the CDF as 
	\begin{align*}
	\Pr[]{\psi_j(v_{i,j'}) \leq t} & = \Pr[]{F_j^{-1} \left( \left( \transformedquantile_{j'}^{(i)} \right)^{q_j^{(i)}} \right) \leq t} = \Pr[]{  \transformedquantile_{j'}^{(i)} \leq F_j(t)^{\frac{1}{q_j^{(i)}}} }\\ & = \Pr[]{F_{j'} (v_{i,j'}) \leq F_j(t)^{\frac{q_{j'}^{(i)}}{q_j^{(i)}}}  } = F_j(t)^{\frac{q_{j'}^{(i)}}{q_j^{(i)}}} \enspace,
	\end{align*} 
	where in the last step, we used that $F_{j'} (v_{i,j'}) \sim \text{Unif}[0,1]$. As a consequence, 
	\begin{align*}
	\Pr{v_{i,j} \leq t \growingmid X_{i,j} = 1, M^{(i)}} & = \Pr{\psi_j(v_{i,j}) \leq t \growingmid \psi_j(v_{i,j}) > \psi_{j'}(v_{i,j}) \text{ for $j \neq j'$}, M^{(i)}}\\
	& = \Pr[]{\max_{j' \in M^{(i)}} \left( \psi_j(v_{i,j'})  \right) \leq t \growingmid M^{(i)}} \\
	& = \prod_{j' \in M^{(i)}} \Pr[]{\psi_j(v_{i,j'}) \leq t} = \prod_{j' \in M^{(i)}} F_j(t)^{\frac{q_{j'}^{(i)}}{q_j^{(i)}}} \\& = F_j(t)^{\frac{\sum_{j' \in M^{(i)}}q_{j'}^{(i)}}{q_j^{(i)}}} = F_j(t)^{\frac{n-i+1}{q_j^{(i)}}} \enspace.
	\end{align*}
\end{proof}

For integral values of $\frac{n-i+1}{q_j^{(i)}}$ (in particular $q_j^{(i)} = 1$), observe that this is exactly the CDF of the maximum of $\frac{n-i+1}{q_j^{(i)}}$ independent draws from distribution $F_j$.

For the remainder of the proof sketch, let us restrict to the case that $m=n$. Observe that in this special case, we have that all $q_j^{(i)} = 1$, so the probability in the quantile allocation rule of allocating any item $j \in M^{(i)}$ simplifies to $\frac{1}{n-i+1}$. Therefore, \[ \Ex{v_{i, j} X_{i, j}} = \frac{n-i+1}{n} \cdot \frac{1}{n-i+1} \cdot \Ex{v_{i, j} \growingmid X_{i, j} = 1 } = \frac{1}{n} \Ex{v_{i, j} \growingmid X_{i, j} = 1 } \enspace. \] Observe that this argument looks rather innocent in the special case of $m=n$, but requires a much more careful treatment in the general variant: The probabilities $q_j^{(i)}$ are random variables themselves depending on the set $M^{(i)}$. Hence, the calculation can not directly be extended and a more sophisticated argument needs to be applied. In addition, by the above considerations on the quantile allocation rule via Lemma \ref{lemma_quantile_allocation}, we have that $\Ex{v_{i, j} \growingmid X_{i, j} = 1 } = \Ex[]{\max_{i' \in \left[n-i+1\right]} v_{i', j} }$ in the special case of $m=n$. Therefore, we can now simply apply Lemma \ref{Lemma:Babaioff_bound_maximum} (see below) to get 
\begin{align*}
\Ex{v_{i, j} X_{i, j}} = \frac{1}{n} \Ex{v_{i, j} \growingmid X_{i, j} = 1 } = \frac{1}{n} \ \Ex[]{\max_{i' \in \left[n-i+1\right]} v_{i', j} } \geq \frac{1}{n} \cdot \frac{H_{n-i+1}}{H_n} \cdot \Ex[]{\max_{i' \in [n]} v_{i', j}}
\end{align*}
Note that we take the maximum over i.i.d.\ random variables. As a consequence, we can conclude by basic calculations:
\begin{align*}
\Ex[]{\ALG} & \geq \Ex[]{\textnormal{SW}_{\textnormal{quantile}}} = \sum_{i=1}^{n} \sum_{j=1}^{n} \Ex[]{v_{i, j} X_{i, j}} \geq \frac{1}{n} \sum_{i=1}^{n} \sum_{j=1}^{n} \frac{H_{n-i+1}}{H_n} \Ex[]{\max_{i' \in [n]} v_{i', j}}
\\ & = \frac{\sum_{i=1}^{n} H_{n-i+1}}{n H_n} \sum_{j=1}^{n} \Ex[]{\max_{i' \in [n]} v_{i', j}}  = \left( 1 - O \left( \frac{1}{\log n} \right) \right) \sum_{j=1}^{n} \Ex[]{\max_{i' \in [n]} v_{i', j}} \\ & \geq \left( 1 - O \left( \frac{1}{\log n} \right) \right)\Ex[]{\OPT}
\end{align*}

Observe that in the general version, comparing to $\sum_{j=1}^{n} \Ex[]{\max_{i' \in [n]} v_{i', j}}$ is a far too strong benchmark. Therefore, we consider an ex-ante relaxation of the offline optimum. As a new technical tool, we introduce in Lemma \ref{lemma:quantile_maximum} (see below) an appropriate bound which allows to lower bound the expected maximum of $k$ draws from an MHR distribution by a suitable fraction of $ \Ex[]{v_{i,j} \growingmid  v_{i,j} \geq F^{-1}\left(1-q\right) }$ for any choice of $q \in [0,1]$. Applying this, we can lower bound the expected contribution of any item $j$ to the quantile welfare by a suitable fraction of its contribution to the offline optimum.

We conclude by stating the lemmas which were used in the proof sketch. First, we restate a useful lemma from \citet{Babaioff:2017:ACM:3037382}. It allows to compare the expectation of the maximum of $n$ and $n' \leq n$ draws from independent and identically distributed random variables, if the distribution has a monotone hazard rate.

\begin{lemma}[Lemma 5.3 in \citet{Babaioff:2017:ACM:3037382}]\label{Lemma:Babaioff_bound_maximum}
	Consider a collection $(X_i)_i$ of independent and identically distributed random variables with a distribution with monotone hazard rate. Then, for any $n' \leq n$, we have \[ \frac{\Ex[]{\max_{i \in [n']} X_i}}{\Ex[]{\max_{i \in [n]} X_i}} \geq \frac{H_{n'}}{H_n} \geq \frac{\log n'}{\log n} \enspace. \] 
\end{lemma}

In addition, we make use of the following lemma which is used in order to make a suitable comparison to the ex-ante relaxation.

\begin{restatable}{lemma}{lemmaquantilemaximum}
	\label{lemma:quantile_maximum}
	Let $\quant \in [0,1]$ and $k \in \mathbb{N}$. Further, let $\mathcal{D}$ be a distribution with monotone hazard rate with CDF $F$, let $X, (Y_i)_i \sim \mathcal{D}$ be independent and identically distributed. For $\alpha \geq \frac{1 + \ln \left(\frac{1}{\quant}\right)}{H_k}$, $\alpha \geq 1$, and $\alpha k \leq \frac{1}{\quant}$, we have $\Ex[]{X \growingmid  X \geq F^{-1}\left(1-\quant\right) } \leq \alpha \Ex[]{\max_{i \in [k]} Y_i } \enspace.$
\end{restatable}

The proof of this lemma can be found in Appendix \ref{appendix:proof_lemma_quantile_maximum}.

\subsection{Static prices}
\label{section:independent_valuations_static}

Next, we would like to demonstrate how to use static prices. We consider the case that the number of items $m$ is upper bounded by $\frac{n}{\left(\log \log n\right)^2}$. We set the price for item $j$ to \[ p_j = F_j^{-1} \left( 1 - q \right) \textnormal{ , where } q = \frac{\log \log n}{n} \ , \] which allows us to state the following theorem.

\begin{restatable}{theorem}{theoremindependentstatic}
	\label{theorem_independent_static}
	The posted-prices mechanism with static prices and independent item-valuations is $1 - O \left( \frac{\log \log \log n}{\log n} \right)$-competitive with respect to social welfare.
\end{restatable}

As before, we defer the proof of this theorem to Appendix~\ref{appendix:independent_valuations_static} and give a quick sketch here: First, observe that we can bound the probability of selling item $j$ to buyer $i$ by the probability of the event that buyer $i$ has only non-negative utility for this item. This implies a bound on the probability of selling item $j$ in our mechanism. Finally, we combine this with a lower bound on the price $p_j$ and hence are able to bound the revenue (and thus the welfare) obtained by our mechanism. Observe that our guarantee only applies if the number of items $m$ is bounded by $\frac{n}{\left(\log \log n\right)^2}$. We leave the extension to the general case as an open problem. As a first step, one could try to derive a suitable bound on the utility of agents in order to extend the result.

	\section{Asymptotically Tight Bounds for Separable Valuations}
\label{section_alpha_correlation}

Let us now come to \emph{separable valuations}, which are common in ad auctions \citep{EdelmanOS07,Varian07}. That is, in order to determine buyer $i$'s value for item $j$, let each buyer $i$ have a type $v_i \geq 0$ and let each item have an item-dependent multiplier $\alpha_j$ which can be interpreted as a click through rate in online advertising. Buyer $i$'s value $v_{i,j}$ for being assigned item $j$ is given by $\alpha_j \cdot v_i$. Without loss of generality, we assume that $\alpha_1 \geq \alpha_2 \geq \ldots$ and that $m = n$. The former can be ensured by reordering the items; the latter by adding items with $\alpha_j = 0$ or removing all items $j \in M$ with $j > n$ respectively. Observe that in the case of $m > n$, items of index larger than $n$ are not matched in either the optimum, nor is it beneficial to match one of these items and leave an item $j \leq n$ unmatched. Note that this correlated setting also contains the single-item scenario as a special case since it can be modeled by $\alpha_1 = 1, \alpha_2 = \ldots = 0$. More generally, $k$ identical items can be modeled by $\alpha_1 = \ldots = \alpha_k = 1, \alpha_{k+1} = \ldots = 0$. 

The types $v_1, \ldots, v_n \geq 0$ are non-negative, independent and identically distributed random variables with a continuous distribution satisfying the MHR condition. Let $j_i$ denote the item allocated to buyer $i$ and $j_i = m+1$ if buyer $i$ is not allocated any item where $\alpha_{m+1} = 0$. We can specify the expected social welfare of the matching computed by the mechanism as $\Ex[]{\sum_{i = 1}^n \alpha_{j_i} v_i }=: \Ex[]{\ALG}$.

Additionally, the structure of the optimal matching can be stated explicitly. Given any type profile $v = (v_1, \ldots, v_n)$, we let $v_{(k)}$ denote the $k$-th highest order statistics. That is, $v_{(k)}$ is the largest $x$ such that there are at least $k$ entries in $v$ whose value is at least $x$. Denote its expectation by $\Ex[]{v_{(k)}} = \mu_k$. The allocation that maximizes social welfare assigns item $1$ to a buyer of type $v_{(1)}$, item $2$ to a buyer of type $v_{(2)}$ and so on. Hence, the expected optimal social welfare is given by $\Ex[]{\OPT} = \Ex{\sum_{j = 1}^m \alpha_j v_{(j)}} = \sum_{j = 1}^m \alpha_j \mu_j$.

\subsection{Dynamic prices}
\label{section:alpha_correlation_dynamic}
First, we focus on posted-prices mechanisms with dynamic prices. Consider step $i$ and buyer $i$ arrives. Let $M^{(i)}$ be the set of remaining items at this time. Our choice of prices ensures that in each step one item is allocated. Therefore, always $\lvert M^{(i)} \rvert = n - i + 1$.

We choose prices $(p^{(i)}_j)_{j \in M^{(i)}}$ with the goal that each item is allocated with probability $\frac{1}{n - i + 1}$. To this end, let $M^{(i)} = \{ \ell_1, \ldots, \ell_{n-i+1} \}$ with $\ell_1 < \ell_2 < \ldots < \ell_{n-i+1}$ and set
\[
p_j^{(i)} = \sum_{k: j \leq \ell_k \leq n-i} (\alpha_{\ell_k} - \alpha_{\ell_{k+1}}) F^{-1}\left( 1 - \frac{k}{n - i + 1}\right) \enspace.
\]
Given this pricing scheme, we can state the following theorem.

\begin{restatable}{theorem}{alphacorrelationdynamic}
	\label{theorem:alpha_correlation_dynamic-positive}
	The posted-prices mechanism with dynamic prices and separable valuations is $1 - O\left( \frac{1}{\log n} \right)$-competitive with respect to social welfare.
\end{restatable}
  
The proof can be found in Appendix~\ref{appendix:correlated_valuations_dynamic}. First, observe that in principle, buyers will be indifferent between two items $j$ and $j'$ if $\alpha_j = \alpha_{j'}$. As these items are indistinguishable for later buyers anyway and new prices will be defined, we can assume that ties are broken in our favor. That is why we can assume that buyer $i$ will prefer item $\ell_k$ if and only if $F^{-1}\left( 1 - \frac{k}{n - i + 1}\right) \leq v_i < F^{-1}\left( 1 - \frac{k-1}{n - i + 1}\right)$. 

Using a suitable lower bound for $F^{-1}\left( 1 - \frac{k}{n - i + 1}\right)$ via the MHR property, we get a lower bound for the value $v_{i,j}$ if $i$ is matched to $j$. To this end, we compare quantiles of MHR distributions to the respective order statistics. As stated before, by Section~\ref{section:optimality}, the competitive ratio is optimal.

\subsection{Static prices}
\label{section:alpha_correlation_static}
When restricting to the case of static prices, we define probabilities $q_j$ having the interpretation that a buyer drawn from the distribution has one of items $1, \ldots, j$ as their first choice. For technical reasons, we discard items $\widehat{m} + 1, \ldots, n$ for $\widehat{m} = n - n^{5/6}$ by setting $p_j = \infty$ for these items. For $j \leq \widehat{m}$, we set prices
\[
p_j = \sum_{k = j}^n (\alpha'_k - \alpha'_{k+1}) F^{-1}(1 - q_k) \enspace, \text{ where } q_k = \min\left\{ \frac{k}{n} 2 \log \log n, \frac{k}{n} + \sqrt{\frac{\log n}{n}} \right\} \enspace,
\]
where $\alpha'_k = \alpha_k$ for $k \leq \widehat{m}$ and $0$ otherwise. 

Note the similarity of this price definition to the payments when applying the VCG mechanism. There, the buyer being assigned item $j$ has to pay $\sum_{k = j}^n (\alpha'_k - \alpha'_{k+1}) v_{(k + 1)}$.

\begin{restatable}{theorem}{alphacorrelationstaticpositive}
	\label{theorem:alpha_correlation_staticpositive}
	The posted-prices mechanism with static prices and separable valuations is $1 - O\left(\frac{\log \log \log n}{\log n}\right)$-competitive with respect to social welfare.
\end{restatable}

The proof of this theorem is deferred to Appendix~\ref{appendix:correlated_valuations_static}. The general steps are as follows. We first show that our prices are fairly low, meaning that the probability of selling all items $1,\dots,\widehat{m}$ is reasonably high. Having this, we decompose the social welfare into utility and revenue. The revenue of our mechanism is bounded in terms of the VCG revenue. To this end, we use that our pricing rule is quantile-based and exploit that the quantiles of any MHR distributions are lower-bounded by suitable fractions of expected order statistics. Talking about utility, we use a link to Myerson's theory and virtual values in order to achieve our desired bound. Again, by Section~\ref{section:optimality}, the competitive ratio is asymptotically tight.

	\section{Asymptotically Upper Bounds on the Competitive Ratios}
\label{section:optimality}
Our competitive ratios are asymptotically tight. In this section we provide matching upper bounds showing optimality. To this end, we consider the case of selling a single item with static and dynamic prices respectively. In any of the two cases, we can achieve asymptotic upper bounds on the competitive ratio of posted prices mechanisms which match our results from the previous sections. In particular, we prove that these bounds hold for any choice of pricing strategy. 

\subsection{Dynamic prices}
\label{section_optimality_dynamic}
We consider the guarantee of our dynamic-pricing mechanisms first. Even with a single item and types drawn from an exponential distribution, the best competitive ratio is $1 - \Omega\left(\frac{1}{\log n}\right)$. We simplify notation by omitting indices when possible.

\begin{proposition}
	\label{proposition:optimality-dynamic}
	Let $v_1, \dots, v_n \in \mathbb{R}_{\geq 0}$ be random variables where each $v_i$ is drawn i.i.d.\ from the exponential distribution with rate $1$, i.e., $v_1, \ldots, v_n \sim \mathrm{Exp}(1)$. For all dynamic prices, the competitive ratio of the mechanism picking the first $v_i$ with $v_i \geq p^{(i)}$ is upper bounded by $1 - \Omega\left(\frac{1}{\log n}\right)$.
\end{proposition}

In order to prove Proposition \ref{proposition:optimality-dynamic}, we use that the expected value of the optimal offline solution (the best value in hindsight) is given by $\Ex[]{\max_{i \in [n]} v_i} = H_n$ \citep{Arnold:2008:FCO:1373327}. Therefore, it suffices to show that the expected value of any dynamic pricing rule is upper bounded by $H_n - c$ for some constant $c > 0$.

To upper-bound the expected social welfare of any dynamic pricing rule, we use the fact that this problem corresponds to a Markov decision process and the optimal dynamic prices are given by\footnote{To the best of our knowledge, this is a folklore result.}
\[
p^{(n)} = 0 \quad \text{ and } \quad p^{(i)} = \Ex[]{\max \{v_{i+1}, p^{(i+1)}\}} \quad \text{ for $i < n$} \enspace.
\]
Furthermore, $p^{(0)}$ is exactly the expected social welfare of this mechanism. Therefore, the following lemma with $k = n$ directly proves our claim.

\begin{restatable}{lemma}{mdprecursion}
	\label{lemma:mdp-recursion}
	Let $v_1, \dots, v_n \in \mathbb{R}_{\geq 0}$ be random variables where each $v_i$ is drawn i.i.d.\ from the exponential distribution $\textnormal{Exp}(1)$. Moreover, let $p^{(n)} = 0$ and $p^{(i)} = \Ex[]{\max \{v_{i+1}, p^{(i+1)}\}}$ for $i < n$. Then, we have $p^{(n-k)} \leq H_k - \frac{1}{8}$ for all $2 \leq k \leq n$.
\end{restatable}

The proof via induction over $k$ is deferred to Appendix~\ref{appendix:optimality_dynamic}. 

\subsection{Static prices}
\label{section_optimality_static}
For static pricing rules, we show that any mechanism is $1 - \Omega\left(\frac{\log \log \log n}{\log n}\right)$-competitive. Again, this bound even holds for a single item and the valuations being drawn from an exponential distribution.

\begin{restatable}{proposition}{optimalitystatic}
	\label{proposition:optimality-static}
	Let $v_1, \dots, v_n \in \mathbb{R}_{\geq 0}$ be random variables where each $v_i$ is drawn i.i.d.\ from the exponential distribution with rate $1$, i.e., $v_1, \ldots, v_n \sim \mathrm{Exp}(1)$. For all static prices $p \in \mathbb{R}_{\geq 0}$ the competitive ratio of the mechanism picking the first $v_i$ with $v_i \geq p$ is upper bounded by $1 - \Omega\left(\frac{\log \log \log n}{\log n}\right)$.
\end{restatable}

The proof of Proposition~\ref{proposition:optimality-static} is deferred to Appendix~\ref{appendix:optimality_static}. The idea is as follows. The expected welfare obtained by the static-price mechanism using price $p$ is given by $\Ex{\ALG} = \Ex[]{v \growingmid v \geq p} \cdot \Pr{\exists i: v_i \geq p} = (p + 1) \cdot \left(1 - \left( 1 - \e^{-p} \right)^n \right)$. This has to be compared to the expected value of the optimal offline solution (the best value in hindsight), which is given by $\Ex[]{\max_{i \in [n]} v_i} = H_n$ \citep{Arnold:2008:FCO:1373327}. 
	\section{Extensions to Subadditive Buyers and Revenue Considerations}
\label{section:subadditive}

In this section, we illustrate that the same style of mechanisms used for unit-demand buyers in principle is also applicable for subadditive buyers. A valuation function $v_i$ is subadditive if $v_i(S \cup T) \leq v_i(S) + v_i(T)$ for any $S, T \subseteq M$. This generalizes unit-demand functions considered so far in this paper. Instead of being interested in only a single-item, each buyer now has a subadditive valuation function over item bundles and can thus be interested in multiple items.

To generalize the MHR property, we assume that the subadditive valuation functions are drawn from distributions with MHR marginals. That is, $v_i \sim \mathcal{D}$ and we assume that $v_i\left( \{ j \} \right)$ has a marginal distribution with monotone hazard rate. Buyers arrive online and purchase the bundle of items which maximizes the buyer's utility.

We can construct a dynamic-pricing mechanism which is $1 - O \left( \frac{1+\log m}{\log n}\right)$-competitive. For a detailed explanation, we refer to Appendix~\ref{appendix:mhr_marginals_dynamic}. The general approach is to split the set of buyers in subgroups of size $\left\lfloor \frac{n}{m} \right\rfloor$ and sell each item to one of these groups. For the $k$-th buyer in every group, the price for the item in question is set to $p_j^{(k)} = F_j^{-1}\left( 1 - \frac{1}{\left\lfloor \frac{n}{m} \right\rfloor - k + 1 } \right)$, where $F_j^{-1}$ denotes the quantile function of the marginal distribution of $v_i\left( \{j\} \right)$. Using techniques similar to the ones in the unit-demand case allows to bound the revenue of the posted-prices mechanism by the desired fraction of the optimal social welfare. Hence, the argument directly implies the respective bounds for welfare and revenue. 

In the static pricing environment, our results can be extended to a mechanism which is $1 - O \left( \frac{\log\log\log n}{\log n} + \frac{\log m}{\log n}\right)$-competitive for subadditive buyers. Details can be found in Appendix~\ref{appendix:mhr_marginals_static}. As before, let $F_j^{-1}$ be the quantile function of the marginal distribution of $v_i\left( \{j\} \right)$. Setting fairly low prices of $p_j = F_j^{-1} \left( 1 - q \right)$ for $q = \frac{m \log \log n}{n}$ ensures that we sell all items with a suitably high probability. Afterwards, we can apply the same bounds for MHR distributions as in the previous sections in order to bound the revenue of the mechanism with the respective fraction of the optimal social welfare. Again, this directly implies the mentioned competitive ratio for welfare as well as revenue, as the revenue of any individually rational mechanism is upper-bounded by the corresponding social welfare. 

Note that the guarantees now depend on the number of items $m$. To make them meaningful, we need $m = o(n)$. This makes them significantly worse than the ones we obtain for unit-demand functions with a much more careful treatment. However, they are stronger in one aspect, namely that in both cases we bound the revenue of the mechanism in terms of the optimal social welfare. In particular, this means that they are also approximations of the optimal revenue. Interestingly, the optimality results from Section~\ref{section:optimality} also transfer.

\subsection*{Acknowledgments}

We thank the anonymous reviewers for helpful comments on improving the presentation of the paper.
	\bibliography{references,dblp}

\begin{thebibliography}{29}
\providecommand{\natexlab}[1]{#1}
\providecommand{\url}[1]{\texttt{#1}}
\expandafter\ifx\csname urlstyle\endcsname\relax
  \providecommand{\doi}[1]{doi: #1}\else
  \providecommand{\doi}{doi: \begingroup \urlstyle{rm}\Url}\fi

\bibitem[Abolhassani et~al.(2017)Abolhassani, Ehsani, Esfandiari, Hajiaghayi,
  Kleinberg, and Lucier]{DBLP:conf/stoc/AbolhassaniEEHK17}
M.~Abolhassani, S.~Ehsani, H.~Esfandiari, M.~Hajiaghayi, R.~D. Kleinberg, and
  B.~Lucier.
\newblock Beating 1-1/e for ordered prophets.
\newblock In H.~Hatami, P.~McKenzie, and V.~King, editors, \emph{Proceedings of
  the 49th Annual {ACM} {SIGACT} Symposium on Theory of Computing, {STOC} 2017,
  Montreal, QC, Canada, June 19-23, 2017}, pages 61--71. {ACM}, 2017.
\newblock \doi{10.1145/3055399.3055479}.
\newblock URL \url{https://doi.org/10.1145/3055399.3055479}.

\bibitem[Alaei(2014)]{DBLP:journals/siamcomp/Alaei14}
S.~Alaei.
\newblock Bayesian combinatorial auctions: Expanding single buyer mechanisms to
  many buyers.
\newblock \emph{{SIAM} J. Comput.}, 43\penalty0 (2):\penalty0 930--972, 2014.
\newblock \doi{10.1137/120878422}.
\newblock URL \url{https://doi.org/10.1137/120878422}.

\bibitem[Arnold et~al.(2008)Arnold, Balakrishnan, and
  Nagaraja]{Arnold:2008:FCO:1373327}
B.~C. Arnold, N.~Balakrishnan, and H.~N. Nagaraja.
\newblock \emph{A First Course in Order Statistics (Classics in Applied
  Mathematics)}.
\newblock Society for Industrial and Applied Mathematics, Philadelphia, PA,
  USA, 2008.
\newblock ISBN 0898716489, 9780898716481.

\bibitem[Babaioff et~al.(2017)Babaioff, Blumrosen, Dughmi, and
  Singer]{Babaioff:2017:ACM:3037382}
M.~Babaioff, L.~Blumrosen, S.~Dughmi, and Y.~Singer.
\newblock Posting prices with unknown distributions.
\newblock \emph{ACM Trans. Econ. Comput.}, 5\penalty0 (2), Mar. 2017.
\newblock ISSN 2167-8375.
\newblock \doi{10.1145/3037382}.
\newblock URL \url{https://doi.org/10.1145/3037382}.

\bibitem[Barlow and Marshall(1964)]{barlow1964}
R.~E. Barlow and A.~W. Marshall.
\newblock Bounds for distributions with monotone hazard rate, i.
\newblock \emph{Ann. Math. Statist.}, 35\penalty0 (3):\penalty0 1234--1257, 09
  1964.
\newblock \doi{10.1214/aoms/1177703281}.
\newblock URL \url{https://doi.org/10.1214/aoms/1177703281}.

\bibitem[Cai and Zhao(2017)]{DBLP:conf/stoc/CaiZ17}
Y.~Cai and M.~Zhao.
\newblock Simple mechanisms for subadditive buyers via duality.
\newblock In H.~Hatami, P.~McKenzie, and V.~King, editors, \emph{Proceedings of
  the 49th Annual {ACM} {SIGACT} Symposium on Theory of Computing, {STOC} 2017,
  Montreal, QC, Canada, June 19-23, 2017}, pages 170--183. {ACM}, 2017.
\newblock \doi{10.1145/3055399.3055465}.
\newblock URL \url{https://doi.org/10.1145/3055399.3055465}.

\bibitem[Chawla et~al.(2007)Chawla, Hartline, and
  Kleinberg]{DBLP:conf/sigecom/ChawlaHK07}
S.~Chawla, J.~D. Hartline, and R.~D. Kleinberg.
\newblock Algorithmic pricing via virtual valuations.
\newblock In J.~K. MacKie{-}Mason, D.~C. Parkes, and P.~Resnick, editors,
  \emph{Proceedings 8th {ACM} Conference on Electronic Commerce (EC-2007), San
  Diego, California, USA, June 11-15, 2007}, pages 243--251. {ACM}, 2007.
\newblock \doi{10.1145/1250910.1250946}.
\newblock URL \url{https://doi.org/10.1145/1250910.1250946}.

\bibitem[Chawla et~al.(2010)Chawla, Hartline, Malec, and
  Sivan]{DBLP:conf/stoc/ChawlaHMS10}
S.~Chawla, J.~D. Hartline, D.~L. Malec, and B.~Sivan.
\newblock Multi-parameter mechanism design and sequential posted pricing.
\newblock In L.~J. Schulman, editor, \emph{Proceedings of the 42nd {ACM}
  Symposium on Theory of Computing, {STOC} 2010, Cambridge, Massachusetts, USA,
  5-8 June 2010}, pages 311--320. {ACM}, 2010.
\newblock \doi{10.1145/1806689.1806733}.
\newblock URL \url{https://doi.org/10.1145/1806689.1806733}.

\bibitem[Chawla et~al.(2017)Chawla, Devanur, Holroyd, Karlin, Martin, and
  Sivan]{DBLP:conf/stoc/ChawlaDHKMS17}
S.~Chawla, N.~R. Devanur, A.~E. Holroyd, A.~R. Karlin, J.~B. Martin, and
  B.~Sivan.
\newblock Stability of service under time-of-use pricing.
\newblock In H.~Hatami, P.~McKenzie, and V.~King, editors, \emph{Proceedings of
  the 49th Annual {ACM} {SIGACT} Symposium on Theory of Computing, {STOC} 2017,
  Montreal, QC, Canada, June 19-23, 2017}, pages 184--197. {ACM}, 2017.
\newblock \doi{10.1145/3055399.3055455}.
\newblock URL \url{https://doi.org/10.1145/3055399.3055455}.

\bibitem[Clarke(1971)]{RePEc:kap:pubcho:v:11:y:1971:i:1:p:17-33}
E.~Clarke.
\newblock Multipart pricing of public goods.
\newblock \emph{Public Choice}, 11\penalty0 (1):\penalty0 17--33, 1971.
\newblock URL
  \url{https://EconPapers.repec.org/RePEc:kap:pubcho:v:11:y:1971:i:1:p:17-33}.

\bibitem[Correa et~al.(2017)Correa, Foncea, Hoeksma, Oosterwijk, and
  Vredeveld]{DBLP:conf/sigecom/CorreaFHOV17}
J.~R. Correa, P.~Foncea, R.~Hoeksma, T.~Oosterwijk, and T.~Vredeveld.
\newblock Posted price mechanisms for a random stream of customers.
\newblock In C.~Daskalakis, M.~Babaioff, and H.~Moulin, editors,
  \emph{Proceedings of the 2017 {ACM} Conference on Economics and Computation,
  {EC} '17, Cambridge, MA, USA, June 26-30, 2017}, pages 169--186. {ACM}, 2017.
\newblock \doi{10.1145/3033274.3085137}.
\newblock URL \url{https://doi.org/10.1145/3033274.3085137}.

\bibitem[Correa et~al.(2019)Correa, Foncea, Pizarro, and
  Verdugo]{DBLP:journals/orl/CorreaFPV19}
J.~R. Correa, P.~Foncea, D.~Pizarro, and V.~Verdugo.
\newblock From pricing to prophets, and back!
\newblock \emph{Oper. Res. Lett.}, 47\penalty0 (1):\penalty0 25--29, 2019.
\newblock \doi{10.1016/j.orl.2018.11.010}.
\newblock URL \url{https://doi.org/10.1016/j.orl.2018.11.010}.

\bibitem[D\"{u}tting et~al.(2020)D\"{u}tting, Kesselheim, and
  Lucier]{10.1145/3440968.3440972}
P.~D\"{u}tting, T.~Kesselheim, and B.~Lucier.
\newblock An o(log log m) prophet inequality for subadditive combinatorial
  auctions.
\newblock \emph{SIGecom Exch.}, 18\penalty0 (2):\penalty0 32–37, Dec. 2020.
\newblock \doi{10.1145/3440968.3440972}.
\newblock URL \url{https://doi.org/10.1145/3440968.3440972}.

\bibitem[Edelman et~al.(2007)Edelman, Ostrovsky, and Schwarz]{EdelmanOS07}
B.~Edelman, M.~Ostrovsky, and M.~Schwarz.
\newblock Selling billions of dollars worth of keywords: The generalized
  second-price auction.
\newblock \emph{American Economic Review}, 97\penalty0 (1):\penalty0 242--259,
  2007.

\bibitem[Ehsani et~al.(2018)Ehsani, Hajiaghayi, Kesselheim, and
  Singla]{DBLP:conf/soda/EhsaniHKS18}
S.~Ehsani, M.~Hajiaghayi, T.~Kesselheim, and S.~Singla.
\newblock Prophet secretary for combinatorial auctions and matroids.
\newblock In A.~Czumaj, editor, \emph{Proceedings of the Twenty-Ninth Annual
  {ACM-SIAM} Symposium on Discrete Algorithms, {SODA} 2018, New Orleans, LA,
  USA, January 7-10, 2018}, pages 700--714. {SIAM}, 2018.
\newblock \doi{10.1137/1.9781611975031.46}.
\newblock URL \url{https://doi.org/10.1137/1.9781611975031.46}.

\bibitem[Feldman et~al.(2015)Feldman, Gravin, and
  Lucier]{DBLP:conf/soda/FeldmanGL15}
M.~Feldman, N.~Gravin, and B.~Lucier.
\newblock Combinatorial auctions via posted prices.
\newblock In P.~Indyk, editor, \emph{Proceedings of the Twenty-Sixth Annual
  {ACM-SIAM} Symposium on Discrete Algorithms, {SODA} 2015, San Diego, CA, USA,
  January 4-6, 2015}, pages 123--135. {SIAM}, 2015.
\newblock \doi{10.1137/1.9781611973730.10}.
\newblock URL \url{https://doi.org/10.1137/1.9781611973730.10}.

\bibitem[Giannakopoulos and Zhu(2018)]{DBLP:conf/wine/GiannakopoulosZ18}
Y.~Giannakopoulos and K.~Zhu.
\newblock Optimal pricing for {MHR} distributions.
\newblock In G.~Christodoulou and T.~Harks, editors, \emph{Web and Internet
  Economics - 14th International Conference, {WINE} 2018, Oxford, UK, December
  15-17, 2018, Proceedings}, volume 11316 of \emph{Lecture Notes in Computer
  Science}, pages 154--167. Springer, 2018.
\newblock \doi{10.1007/978-3-030-04612-5\_11}.
\newblock URL \url{https://doi.org/10.1007/978-3-030-04612-5\_11}.

\bibitem[Groves(1973)]{RePEc:ecm:emetrp:v:41:y:1973:i:4:p:617-31}
T.~Groves.
\newblock Incentives in teams.
\newblock \emph{Econometrica}, 41\penalty0 (4):\penalty0 617--31, 1973.
\newblock URL
  \url{https://EconPapers.repec.org/RePEc:ecm:emetrp:v:41:y:1973:i:4:p:617-31}.

\bibitem[Hajiaghayi et~al.(2007)Hajiaghayi, Kleinberg, and
  Sandholm]{DBLP:conf/aaai/HajiaghayiKS07}
M.~T. Hajiaghayi, R.~D. Kleinberg, and T.~Sandholm.
\newblock Automated online mechanism design and prophet inequalities.
\newblock In \emph{Proceedings of the Twenty-Second {AAAI} Conference on
  Artificial Intelligence, July 22-26, 2007, Vancouver, British Columbia,
  Canada}, pages 58--65. {AAAI} Press, 2007.
\newblock URL \url{http://www.aaai.org/Library/AAAI/2007/aaai07-009.php}.

\bibitem[Hill et~al.(1982)Hill, Kertz, et~al.]{HillKertz82}
T.~P. Hill, R.~P. Kertz, et~al.
\newblock Comparisons of stop rule and supremum expectations of iid random
  variables.
\newblock \emph{The Annals of Probability}, 10\penalty0 (2):\penalty0 336--345,
  1982.

\bibitem[Jin et~al.(2019)Jin, Li, and Qi]{10.1007/978-3-030-35389-6_17}
Y.~Jin, W.~Li, and Q.~Qi.
\newblock On the approximability of simple mechanisms for mhr distributions.
\newblock In I.~Caragiannis, V.~Mirrokni, and E.~Nikolova, editors, \emph{Web
  and Internet Economics}, pages 228--240, Cham, 2019. Springer International
  Publishing.
\newblock ISBN 978-3-030-35389-6.

\bibitem[Krengel and Sucheston(1977)]{krengel1977semiamarts}
U.~Krengel and L.~Sucheston.
\newblock Semiamarts and finite values.
\newblock \emph{Bull. Amer. Math. Soc}, 83\penalty0 (4), 1977.

\bibitem[Lucier(2017)]{DBLP:journals/sigecom/Lucier17}
B.~Lucier.
\newblock An economic view of prophet inequalities.
\newblock \emph{SIGecom Exchanges}, 16\penalty0 (1):\penalty0 24--47, 2017.
\newblock \doi{10.1145/3144722.3144725}.
\newblock URL \url{https://doi.org/10.1145/3144722.3144725}.

\bibitem[Myerson(1981)]{DBLP:journals/mor/Myerson81}
R.~B. Myerson.
\newblock Optimal auction design.
\newblock \emph{Math. Oper. Res.}, 6\penalty0 (1):\penalty0 58--73, 1981.
\newblock \doi{10.1287/moor.6.1.58}.
\newblock URL \url{https://doi.org/10.1287/moor.6.1.58}.

\bibitem[Rinne(2014)]{rinne2014hazard}
H.~Rinne.
\newblock \emph{The Hazard rate : Theory and inference (with supplementary
  MATLAB-Programs)}.
\newblock Justus-Liebig-Universität, 2014.
\newblock URL \url{http://geb.uni-giessen.de/geb/volltexte/2014/10793}.

\bibitem[Samuel-Cahn(1984)]{Samuel84}
E.~Samuel-Cahn.
\newblock Comparison of threshold stop rules and maximum for independent
  nonnegative random variables.
\newblock \emph{Annals of Probability}, 12:\penalty0 1213--1216, 1984.

\bibitem[Varian(2007)]{Varian07}
H.~Varian.
\newblock Position auctions.
\newblock \emph{International Journal of Industrial Organization}, 25\penalty0
  (6):\penalty0 1163–1178, 2007.

\bibitem[Vickrey(1961)]{RePEc:bla:jfinan:v:16:y:1961:i:1:p:8-37}
W.~Vickrey.
\newblock Counterspeculation, auctions, and competitive sealed tenders.
\newblock \emph{Journal of Finance}, 16\penalty0 (1):\penalty0 8--37, 1961.
\newblock URL
  \url{https://EconPapers.repec.org/RePEc:bla:jfinan:v:16:y:1961:i:1:p:8-37}.

\bibitem[Zhang(2020)]{zhang:LIPIcs:2020:12948}
H.~Zhang.
\newblock {Improved Prophet Inequalities for Combinatorial Welfare Maximization
  with (Approximately) Subadditive Agents}.
\newblock In F.~Grandoni, G.~Herman, and P.~Sanders, editors, \emph{28th Annual
  European Symposium on Algorithms (ESA 2020)}, volume 173 of \emph{Leibniz
  International Proceedings in Informatics (LIPIcs)}, pages 82:1--82:17,
  Dagstuhl, Germany, 2020. Schloss Dagstuhl--Leibniz-Zentrum f{\"u}r
  Informatik.
\newblock ISBN 978-3-95977-162-7.
\newblock \doi{10.4230/LIPIcs.ESA.2020.82}.
\newblock URL \url{https://drops.dagstuhl.de/opus/volltexte/2020/12948}.

\end{thebibliography}
	\newpage \appendix
	\section{Proofs for Section \ref{section:independent_valuations_dynamic} on Independent Valuations}
\label{appendix:independent_valuations_dynamic_general}

We give a full proof of Theorem \ref{theorem_independent_dynamic}.

\theoremindependentdynamic*

We set $q_{j}^{(i)}$ to be the probability that item $j$ is matched in the optimum on $M^{(i)}$ and $n-i+1$ buyers if $j \in M^{(i)}$ and $0$ else. Observe that $q_{j}^{(i)}$ are random variables. 

Further, note that we can assume without loss of generality that $m \geq n$. In the case that $m < n$, we can add dummy items with value $0$ to ensure $m=n$. Matching $i$ to one of these dummy items in the mechanism then corresponds to leave $i$ unmatched. Observe that from a technical point of view, a point mass on $0$ is not a MHR distribution and hence, the lemmas for MHR distributions do not apply. Still, for dummy items, $\Ex[]{\max_{i' \in \left[n-i+1\right]} v_{i', j} } = \Ex[]{\max_{i' \in [n]} v_{i', j}}= 0$ which lets us resolve the problem and apply the quantile allocation rule combined with the proof sketch from Section \ref{section:independent_valuations_dynamic} for $m = n$. In particular, in the presence of dummy items, we simply draw an independent sample from $\text{Unif}[0,1]$ for each dummy item in the quantile allocation rule and use this as an artificial quantile. By this, we ensure that each item is still allocated with probability $\frac{1}{n-i+1}$ and further one item (maybe a dummy item) is sold for sure.

\begin{proof}[Proof of Theorem \ref{theorem_independent_dynamic}, general version]
	We lower-bound the expected social welfare $\Ex{\ALG}$ of our posted-prices algorithm by the expected social welfare $\Ex[]{\textnormal{SW}_{\textnormal{quantile}}}$ of the following \emph{generalized quantile allocation rule} stated in Algorithm \ref{algorithm_dynamic_independent_3}:\\
	Consider buyer $i$ in step $i$. Instead of allocating buyer $i$ the item that maximizes her utility $v_{i, j} - p_j^{(i)}$ for $j \in M^{(i)}$, we allocate the item that attains the highest weighted quantile $\max_{j \in M^{(i)}} F_j(v_{i, j})^{\frac{1}{q_j^{(i)}}}$.
	
	As described in Section~\ref{section:independent_valuations_dynamic}, we have $\Ex[]{\textnormal{SW}_{\textnormal{quantile}}} \leq \Ex{\ALG}$. 
	
	\begin{algorithm}[h]
		\SetAlgoNoLine
		\DontPrintSemicolon
		$M^{(1)} \longleftarrow M$ \\
		\For{$i \in N$}{
			\text{Compute $F_j(v_{i, j})^{\frac{1}{q_j^{(i)}}}$ for any $j \in M^{(i)}$ with $q_{j}^{(i)} > 0$} \\ 
			\text{Set $j_i$ such that it attains $\max_{j \in M^{(i)}} F_j(v_{i, j})^{\frac{1}{q_j^{(i)}}}$ and remove $j_i$ from $M^{(i+1)}$ } 
		}
		\caption{Quantile Allocation Rule for Independent Item Valuations}
		\label{algorithm_dynamic_independent_3}
	\end{algorithm}
	
	We define $\transformedquantile_{j}^{(i)} = F_j(v_{i, j})^{\frac{1}{q_j^{(i)}}}$. Observe that now buyer $i$ is allocated item $j$ for which $\transformedquantile_{j}^{(i)} > \transformedquantile_{j'}^{(i)}$ for any $j' \neq j$.
	
	We would like to gain control on the distribution of $\transformedquantile_{j}^{(i)}$. To this end, observe that for any $i$, any $j$ and any $t \in [0,1]$,
	\begin{align*}
	\Pr[]{\transformedquantile_{j}^{(i)} \leq t} = \Pr[]{F_j(v_{i,j}) \leq t^{q_j^{(i)}}} = \Pr[]{v_{i,j} \leq F_j^{-1}\left( t^{q_j^{(i)}} \right)} = F_j \left( F_j^{-1} \left( t^{q_j^{(i)}} \right) \right) = t^{q_j^{(i)}} \enspace.
	\end{align*}
	As a side remark, for $q_j^{(i)} = 1$, this is exactly the CDF of a random variable drawn from $\text{Unif}[0,1]$. \\
	
	Next, observe that in the optimum, we will match any buyer to an item. Therefore, $\sum_{j \in M^{(i)}} q_j^{(i)} = n-i+1$ for any step $i$ as this is the size of the optimum matching on $M^{(i)}$ and $n-i+1$ buyers.
	
	Define indicator variables $X_{i,j}$ which are $1$ if buyer $i$ is allocated item $j$ in the generalized quantile allocation rule, i.e. if $\transformedquantile_{j}^{(i)} > \transformedquantile_{j'}^{(i)}$ for all $j' \in M^{(i)} \setminus \{j\}$.
	
	Our overall goal is to show the following bound.
	
	\begin{proposition}
		Let $\text{OPT}_j$ denote the random variable indicating the contribution of item $j$ to the social welfare of the optimal offline solution. Then
		\[
		\Ex[]{\sum_{i=1}^n v_{i,j} X_{i,j}  } \geq \left( 1 - O\left( \frac{1}{\log n} \right)\right) \Ex[]{\text{OPT}_j}
		\]
	\end{proposition}
	
	Note that showing this proposition proves the claim by taking a sum over all $j \in M$.
	
	To prove the proposition, our first step is to control the distribution of $v_{i,j}$ given that $X_{i,j} = 1$ in order to get access to the value of an agent being allocated an item in the quantile allocation rule. To this end, we restate the corresponding lemma from Section \ref{section:independent_valuations}.
	\lemmaquantileallocation* 
	The proof can be found in Section \ref{proof_lemma_quantile_allocation}. Again, for integral values of $\frac{n-i+1}{q_j^{(i)}}$, observe that this is exactly the CDF of the maximum of $\frac{n-i+1}{q_j^{(i)}}$ independent draws from distribution $F_j$.
		
	Further, we argue on the random variables $q_{j}^{(i)}$. As a first remark, the variables $q_{j}^{(i)}$ are independent of the values $v_{i,j}$ as we define the $q_{j}^{(i)}$ without any knowledge on the $v_{i,j}$. The following observation shows that $\Ex[]{\frac{q_{j}^{(i)}}{n-i+1}}$ is exactly $\frac{q_j^{(1)}}{n}$. Note that $q_j^{(1)}$ is deterministic because it is the a priori probability that item $j$ is matched in the optimum.
	
	\begin{observation}
		For all $i$, we have $\Ex[]{\frac{q_{j}^{(i)}}{n-i+1} } = \frac{q_j^{(1)}}{n}$.
	\end{observation}
	
	\begin{proof}
		Observe that this is the expected probability that we allocate item $j$ in step $i$ of the generalized quantile allocation rule. %To this end, fix the optimal offline solution on valuation profile $(v_{i,j})_{i,j}$ and l
		Let $M_{\ast}^{(i)}$ be the set of items not allocated buyers $1, \ldots, i-1$ by the optimal offline solution. Note that $M^{(i)}$ and $M_{\ast}^{(i)}$ are identically distributed. Therefore, the optimal offline solution assigns item $j$ to one of the buyers $i, \ldots, n$ with a probability of $\Ex{q_{j}^{(i)}}$ a priori. By symmetry between buyers, each buyer is assigned item $j$ in the optimal offline solution with the same probability, namely $\frac{1}{n-i+1} \Ex{q_{j}^{(i)}} = \frac{1}{n} q_j^{(1)}$.
	\end{proof}
	
	\begin{observation}
		\label{observation:qsmaller}
		We have $q_{j}^{(i)} \leq q_{j}^{(1)}$ for all $i$ and $j$.
	\end{observation}
	
	\begin{proof}
		If $j \not\in M^{(i)}$, we have $q_{j}^{(i)} = 0$, so the statement is clear. Otherwise, recall that $q_{j}^{(i)}$ is the probability that item $j$ is allocated in the offline optimum constrained to buyers $1, \ldots, i-1$ receiving the items from $M \setminus M^{(i)}$. Compare for a particular realization of the valuations this constrained to the unconstrained offline optimum. If in the constrained optimum item $j$ is allocated, it is also allocated in the unconstrained optimum. 
	\end{proof}
	
	Now, we first consider the case of step $i$ conditioned on $M^{(i)}$. First, observe that conditioned on $M^{(i)}$, $q_{j}^{(i)}$ is not random anymore. Hence, by the use of Observation \ref{observation:allocation_probabilities} we get that 
	\begin{align*}
	\Ex[]{v_{i,j} X_{i,j}   \growingmid   M^{(i)} } = \Pr[]{X_{i,j} = 1   \growingmid   M^{(i)}} \cdot \Ex[]{v_{i,j}  \growingmid   X_{i,j} = 1, M^{(i)}} = \frac{q_{j}^{(i)}}{n-i+1} \Ex[]{Y_{i,j}} \enspace,
	\end{align*}
	where $Y_{i,j}$ denotes a random variable with CDF $F_j(t)^{\frac{n-i+1}{q_j^{(i)}}}$. By Observation~\ref{observation:qsmaller}, we can bound $\Ex[]{Y_{i,j}} \geq \Ex[]{Y'_{i,j}}$ where $Y'_{i,j}$ is a random variable with CDF $F_j(t)^{\frac{n-i+1}{q_j^{(1)}}}$. Note that the latter CDF does not depend on $q_{j}^{(i)}$ anymore, but only on $q_{j}^{(1)}$ which is deterministic. As a consequence, we can bound \[ 	\Ex[]{v_{i,j} X_{i,j}   \growingmid   M^{(i)} } \geq \frac{q_{j}^{(i)}}{n-i+1} \Ex[]{Y'_{i,j}} \enspace. \]  
	Taking the expectation over all possible sets $M^{(i)}$, we get \[  \Ex[]{v_{i,j} X_{i,j}} \geq \Ex[]{ \frac{q_{j}^{(i)}}{n-i+1} \Ex[]{Y'_{i,j} } }  =  \frac{q_{j}^{(1)}}{n} \Ex[]{Y'_{i,j} } \enspace. \] Now, we can sum over all buyers $i$ to bound the contribution of one item to the quantile allocation rule. Observe that $Y'_{i,j}$ has CDF $F_j(t)^{\frac{n-i+1}{q_j^{(1)}}}$. When rounding the value of $\frac{n-i+1}{q_j^{(1)}}$ to the next smaller integer, we get a random variable which is the maximum of $\left\lfloor \frac{n-i+1}{q_j^{(1)}} \right\rfloor$ draws from $F_j$. Hence, we get
	\begin{align*}
	\Ex[]{\sum_{i=1}^n v_{i,j} X_{i,j}  } &  \geq \sum_{i=1}^{n} \frac{q_{j}^{(1)}}{n} \Ex[]{Y'_{i,j}} \geq  \sum_{i=1}^{n} \frac{q_{j}^{(1)}}{n} \Ex[]{ \max_{i' \in \left\lfloor \frac{n-i+1}{q_j^{(1)}} \right\rfloor} \left\{ v_{i',j} \right\} } 
	\end{align*}
	
	Now, define $k = \left\lceil \frac{n}{2 q_j^{(1)}} \right\rceil$ and $k_{i} = \left\lfloor \frac{n-i+1}{q_j^{(1)}} \right\rfloor$. We make a case distinction if $k_i \geq k$ or not. Denote by $i^\ast$ the last $i$ for which $k_i \geq k$. \\
	
	\textbf{Case 1:} $k_i \geq k$. 
	
	Note that for all $i$ with $k_i = \left\lfloor \frac{n-i+1}{q_j^{(1)}} \right\rfloor \geq k$, i.e. $i = 1 , \dots, i^\ast$, we get \[ \Ex[]{ \max_{i' \in [k_{i}]} \left\{ v_{i',j} \right\} } \geq \Ex[]{ \max_{i' \in [k]} \left\{ v_{i',j} \right\} } \geq \frac{\log(n-i)}{\log(n)} \Ex[]{ \max_{i' \in [k]} \left\{ v_{i',j} \right\} } \enspace. \]
	
	\textbf{Case 2:} $k_i < k$. 
	
	For all $i$ such that $\left\lfloor \frac{n-i+1}{q_j^{(1)}} \right\rfloor < k$, we can exploit the MHR property via an application of Lemma \ref{Lemma:Babaioff_bound_maximum} in order to bound 
	\[
	\Ex[]{ \max_{i' \in \left\lfloor \frac{n-i+1}{q_j^{(1)}} \right\rfloor} \left\{ v_{i',j} \right\} } \geq \frac{\log\left( \left\lfloor \frac{n-i+1}{q_j^{(1)}} \right\rfloor  \right)}{\log\left( k \right) } \Ex[]{ \max_{i' \in [k]} \left\{ v_{i',j} \right\} } \geq \frac{\log\left( n-i \right) - \log\left( q_j^{(1)} \right)}{ \log \left( n+2\right) -  \log \left( 2 q_j^{(1)} \right)} \Ex[]{ \max_{i' \in [k]} \left\{ v_{i',j} \right\} }  \enspace.
	\]
	
	For the last inequality, we use that $\log (k) = \log \left( \left\lceil \frac{n}{2 q_j^{(1)}} \right\rceil \right) \leq \log \left( \frac{n}{2 q_j^{(1)}} + 1 \right)  = \log \left( \frac{n+ 2 q_j^{(1)}}{2q_j^{(1)}} \right) =  \log \left( n+ 2q_j^{(1)}\right) -  \log \left(  2q_j^{(1)} \right) \leq \log \left( n+2\right) -  \log \left( 2 q_j^{(1)} \right)$. By similar calculations, we get that $\log\left( \left\lfloor \frac{n-i+1}{q_j^{(1)}} \right\rfloor  \right) \geq \log\left( n-i \right) - \log\left( q_j^{(1)} \right)$.
	
	Observe, that for $i = n$, the latter expression is not be well defined. We will not consider these variables and only take the sum until $n-1$ later. \\
	
	\textbf{Combination:}
	
	Now, we aim to apply Lemma \ref{lemma:quantile_maximum} for suitably chosen values of $\alpha$, $k$ and $\quant$. We set $\quant = \frac{q_j^{(1)}}{n}$ and $k = \left\lceil \frac{1}{2\quant}\right\rceil = \left\lceil \frac{n}{2 \cdot q_j^{(1)}} \right\rceil$ as above and $\alpha = \frac{1 + \ln \left(n\right)}{H_{n/2}}$. 
	As a consequence, we can compute 
	\begin{align*}
	\Ex[]{\sum_{i=1}^n v_{i,j} X_{i,j}  }  & \geq \sum_{i=1}^{i^\ast} \frac{q_{j}^{(1)}}{n} \frac{\log(n-i)}{\log(n)} \Ex[]{ \max_{i' \in [k]} \left\{ v_{i',j} \right\} } \\ & + \sum_{i=i^\ast + 1}^{n-1} \frac{q_{j}^{(1)}}{n} \frac{\log\left( n-i \right) - \log\left( q_j^{(1)} \right)}{ \log \left( n+2\right) -  \log \left( 2 q_j^{(1)} \right)} \Ex[]{ \max_{i' \in [k]} \left\{ v_{i',j} \right\} } \\ & \geq \sum_{i=1}^{i^\ast} \frac{q_{j}^{(1)}}{n} \frac{1}{\alpha} \frac{\log(n-i)}{\log(n)}  \Ex[]{v_{j}  \growingmid v_{j} \geq F_j^{-1}\left(1-\frac{q_j^{(1)}}{n}\right) } \\ & + 
	\sum_{i=i^\ast + 1}^{n-1} \frac{q_{j}^{(1)}}{n} \frac{1}{\alpha}  \frac{\log\left( n-i \right) - \log\left( q_j^{(1)} \right)}{ \log \left( n+2\right) -  \log \left( 2 q_j^{(1)} \right)} \Ex[]{v_{j} \growingmid  v_{j} \geq F_j^{-1}\left(1-\frac{q_j^{(1)}}{n}\right) }
	\end{align*} 
			
	Now, observe that $q_{j}^{(1)}\Ex[]{v_{j} \growingmid  v_{j} \geq F_j^{-1}\left(1-\frac{q_j^{(1)}}{n}\right) } \geq \Ex[]{\text{OPT}_j}$ because the former is exactly the contribution of item $j$ to the ex-ante relaxation and hence it is an upper bound for the expected contribution of item $j$ to the offline optimum. In addition, we use that by the integral estimation $\sum_{i = 1}^n \log(i) = \sum_{i = 1}^n \log(n-i+1) \geq n \log n - n + 1$. Therefore,
	\begin{align*}
	\Ex[]{\sum_{i=1}^n v_{i,j} X_{i,j}  }  & \geq \Ex[]{\text{OPT}_j} \alpha^{-1} \left( \sum_{i=1}^{i^\ast} \frac{1}{n} \cdot \frac{\log(n-i)}{\log(n)} + \sum_{i=i^\ast + 1}^{n-1} \frac{1}{n} \cdot \frac{\log\left( n-i \right) - \log\left( q_j^{(1)} \right)}{ \log \left( n+2\right) -  \log \left( 2 q_j^{(1)} \right)} \right) \\& \geq \left( 1 - O\left( \frac{1}{\log n} \right)\right) \Ex[]{\text{OPT}_j}
	\end{align*}
\end{proof}

\lemmaquantilemaximum*

\begin{proof}\label{appendix:proof_lemma_quantile_maximum}
	First, we define for $y \in [0,1]$ by $g(y) = \frac{1-y}{f(F^{-1}(y))}$ the inverse of the hazard rate at point $F^{-1}(y)$, i.e. for hazard rate $h(x) = \frac{f(x)}{1-F(x)}$, we have $h\left(F^{-1}(y)\right) = \frac{f(F^{-1}(y))}{1-F\left( F^{-1}(y) \right)} = \frac{f(F^{-1}(y))}{1-y} = \frac{1}{g(y)}$. Observe that by the MHR property, $h(x)$ is non-decreasing and hence $g(y)$ is non-increasing.
	
	Second, we aim for a suitable expression of the $\Ex[]{X \growingmid  X \geq F^{-1}\left(1-\quant\right) }$ which we can compute as follows:
	\begin{align*}
	\Ex[]{X \growingmid  X \geq F^{-1}\left(1-\quant\right) } - F^{-1} (1-\quant)  & = \frac{1}{\quant} \int_{F^{-1} (1-\quant) }^{\infty} 1- F(x) dx \\ & = \frac{1}{\quant} \int_{F^{-1} (1-\quant) }^{\infty} \frac{1- F(x)}{f(x)} f(x) dx \\ & = \frac{1}{\quant} \int_{F^{-1} (1-\quant) }^{\infty} \frac{1}{h(x)} f(x) dx \\ &= \frac{1}{\quant} \int_{1-\quant}^{1} g(y) dy
	\end{align*}
	In addition, observe that we can write $F^{-1} (1-\quant) = \int_{0}^{1-\quant} \frac{g(y)}{1-y} dy$.
	
	Third, note that we can calculate
	\begin{align*}
	\Ex[]{\max_{i \in [k]} Y_i } & = \int_{0}^{\infty} 1 - F^k(x) dx = \int_{0}^{\infty} \frac{1 - F^k(x)}{f(x)} f(x) dx \\ &% = \int_{0}^{\infty} \frac{1 - F(x)}{f(x)} \left( \sum_{i=0}^{k-1} F^{i}(x) \right)f(x) dx 
	= \int_{0}^{\infty} \frac{1}{h(x)} \frac{1 - F^k(x)}{1 - F(x)} f(x) dx \\ & = \int_0^1 g(y) \frac{1-y^k}{1-y} dy \enspace.
	\end{align*} 
	
	As a consequence, the claim of our lemma holds if and only if \[ \alpha \int_0^1 g(y) \frac{1-y^k}{1-y} dy - \int_{0}^{1-\quant} \frac{g(y)}{1-y} dy - \int_{1-\quant}^{1} \frac{g(y)}{\quant} dy \geq 0 \enspace. \]  We can split the left-hand side in the following (possibly empty) integrals: First, split the first integral into two ranges from $0$ to $1-\quant$ and the remainder starting from $1-\quant$. Then, combine the respective integrals over equal ranges and define a threshold $y^\ast = \min \left\{ \sqrt[k]{1 - \frac{1}{\alpha}} , 1-\quant \right\}$ as the point at which the sign of $\alpha \left( 1-y^k \right) -1$ changes from positive to negative. This allows to rewrite the integrals of the left-hand side as \[ \int_0^{y^\ast} \frac{g(y)}{1-y} \left( \alpha \left( 1-y^k \right) - 1 \right) dy + \int_{y^\ast}^{1-\quant} \frac{g(y)}{1-y} \left( \alpha \left( 1-y^k \right) - 1 \right) dy + \int_{1-\quant}^{1} g(y) \left( \alpha \sum_{i=0}^{k-1} y^i - \frac{1}{\quant} \right) dy   \enspace. \] Observe that $g(y)$ is non-increasing by the MHR property. Further, note that $\alpha \sum_{i=0}^{k-1} y^i - \frac{1}{\quant} \leq 0$ as $\alpha k \leq \frac{1}{\quant}$. Setting $c = g(y^\ast)$, we can compute 
	
	\begin{align*}
	&\int_0^{y^\ast} \frac{g(y)}{1-y} \left( \alpha \left( 1-y^k \right) - 1 \right) dy + \int_{y^\ast}^{1-\quant} \frac{g(y)}{1-y} \left( \alpha \left( 1-y^k \right) - 1 \right) dy + \int_{1-\quant}^{1} g(y) \left( \alpha \sum_{i=0}^{k-1} y^i - \frac{1}{\quant} \right) dy  \\  \geq & \int_0^{y^\ast} \frac{c}{1-y} \left( \alpha \left( 1-y^k \right) - 1 \right) dy + \int_{y^\ast}^{1-\quant} \frac{c}{1-y} \left( \alpha \left( 1-y^k \right) - 1 \right) dy + \int_{1-\quant}^{1} c \left( \alpha \sum_{i=0}^{k-1} y^i - \frac{1}{\quant} \right) dy \\  = &  \ c \alpha \int_{0}^{1} \sum_{i=0}^{k-1} y^i dy - c \int_{0}^{1-\quant} \frac{1}{1-y} dy - c \int_{1-\quant}^{1} \frac{1}{\quant} dy = c \alpha H_k - c \left( 1 + \ln\left( \frac{1}{\quant}\right) \right)
	\end{align*}
	
	By our choice of $\alpha \geq \frac{1 + \ln \left(\frac{1}{\quant}\right)}{H_k}$, observe that $\alpha H_k - \left( 1 + \ln \left(\frac{1}{\quant}\right) \right) \geq 0$. So the integral is non-negative for any $c \geq 0$. As a consequence, the claim holds.
\end{proof}

	\section{Proofs for Section \ref{section:independent_valuations_static} on Independent Valuations}
\label{appendix:independent_valuations_static}

\theoremindependentstatic*
%The posted-prices mechanism with static prices, unit-demand buyers and independent item-valuations is $1 - O \left( \frac{\log \log \log n}{\log n} \right)$-competitive.

\begin{proof}
	We start by considering the probability of selling a fixed item $j$. We can lower bound the probability that buyer $i$ buys the item by the event that buyer $i$ only has positive utility for item $j$, i.e.
	\begin{align*}
		\Pr[]{\textnormal{$i$ buys item $j$}} & \geq \left( 1 - F_j(p_j) \right) \prod_{j' \in M^{(i)} \setminus \{j\}} F_{j'} (p_{j'}) \geq \frac{\log \log n}{n} \left( 1- \frac{\log \log n}{n} \right)^{m} \\ & \geq \frac{\log \log n}{n} \left( 1- m \cdot \frac{\log \log n}{n} \right)  \geq \frac{\log \log n}{n} \left( 1- \frac{1}{\log \log n} \right) \enspace,
	\end{align*}
	where the third inequality is an application of Bernoulli's inequality $\left(1+x\right)^r \geq 1+xr$ for any $x \geq -1$ and integer $r$ and the last inequality follows as $m \leq \frac{n}{\left(\log \log n\right)^2}$. Hence, the probability that buyer $i$ does not buy item $j$ is upper bounded by \[ \Pr[]{\textnormal{$i$ does not buy item $j$}} \leq 1 - \frac{\log \log n - 1}{n} \enspace. \] As a consequence, the probability that item $j$ is not sold during the process is upper bounded by \[ \Pr[]{\textnormal{$j$ unsold}} \leq \left( 1 - \frac{\log \log n - 1}{n} \right)^n \leq \exp\left(1-\log \log n\right) = \frac{\e}{ \log n } \enspace. \] Therefore, the probability of selling item $j$ is lower bounded by $\Pr[]{\textnormal{$j$ sold}} \geq 1 - \frac{\e}{ \log n }$. Additionally, we can bound the price of item $j$ by Lemma \ref{lemma:quantiles1} in the following way:
	\begin{align*}
		p_j = F_j^{-1} \left( 1 - \frac{\log \log n}{n} \right) \geq \frac{\log n - \log \log \log n}{H_n} \Ex[]{\max_{i \in [n]} v_{i,j} } \geq \left( 1 - \frac{\log \log \log n + 1}{\log n} \right) \Ex[]{\max_{i \in [n]} v_{i,j} }
	\end{align*}
	Having this, we can conclude by the some fundamental calculus.
	\begin{align*}
		\Ex[]{\ALG} & \geq \Ex[]{\revpp} \geq \sum_{j=1}^{m} \Pr[]{\textnormal{$j$ sold}} \cdot p_j \\ & \geq \sum_{j=1}^{m} \left(1 - \frac{\e}{ \log n }  \right) \left( 1 - \frac{\log \log \log n + 1}{\log n} \right) \Ex[]{\max_{i \in [n]} v_{i,j} }  \\ & \geq \left( 1 - O \left( \frac{\log \log \log n}{\log n} \right) \right) \sum_{j=1}^{m} \Ex[]{\max_{i \in [n]} v_{i,j} } \geq \left( 1 - O \left( \frac{\log \log \log n}{\log n} \right) \right) \Ex[]{\OPT}
	\end{align*}
	
\end{proof}

	\section{Proofs for Section \ref{section:alpha_correlation_dynamic} on Separable Valuations}
\label{appendix:correlated_valuations_dynamic}

Here, we give a full proof of Theorem \ref{theorem:alpha_correlation_dynamic-positive}.

\alphacorrelationdynamic*
%The posted-prices mechanism with dynamic prices, unit-demand buyers and alpha-correlated valuations is $1 - O\left( \frac{1}{\log n} \right)$-competitive. \\

We start by proving some helpful lemmas.

\begin{lemma}
	\label{lemma:quantiles2}
	For all $q \in [0, 1]$, we have
	\[
	F^{-1}(1 - q) \geq \mu_{\left\lfloor nq + \sqrt{n \log n} \right\rfloor} \enspace.
	\]
\end{lemma}

\begin{proof}
	Let $y = F^{-1}(1 - q)$ and $k = \left\lfloor nq + \sqrt{n \log n} \right\rfloor$. Define independent Bernoulli random variables $Z_i$ by setting $Z_i = 1$ if $v_i \geq y$ and $0$ otherwise. By this definition, $\Ex{Z_i} = q$. Furthermore, define $Z = \sum_{i = 1}^n Z_i$. Observe that $v_{(k)} < y$ if and only if $\sum_{i = 1}^n Z_i < k$.
	
	By Hoeffding's inequality, using $\Ex{Z} = n q$, we have
	\[
	\Pr{Z \geq k} = \Pr{Z \geq nq + \sqrt{n \log n}} \leq \exp\left( - 2 \frac{n \log n}{n} \right) = \frac{1}{n^2} \enspace.
	\]
	And therefore
	\[
	\Pr{v_{(k)} < y} = \Pr{Z < k} \geq 1 - \frac{1}{n^2} \enspace.
	\]

	Furthermore, as $v_{(k)}$ is an order statistic, it is distributed according to an MHR distribution \citep{rinne2014hazard}. Therefore (cf. \citep{barlow1964}, Theorem 3.8)
	\[
	\Pr{v_{(k)} \leq \mu_k} \leq 1 - \frac{1}{\e} \enspace.
	\]
	Therefore, since $n^2 \geq \e$, we have to have $\mu_k \leq y$.
\end{proof}

Next, we reformulate a useful lemma from \citet{DBLP:conf/wine/GiannakopoulosZ18}.

\begin{lemma}
	\label{lemma:quantiles1}
	For all $q$ and $j$ such that $\exp(H_{j-1} - H_n) \leq q \leq 1$, we have
	\[
	F^{-1}(1 - q) \geq - \frac{\log(q)}{H_n - H_{j-1}} \mu_j \enspace.
	\]
\end{lemma}

\begin{proof}
	Use Lemma 3 in \citet{DBLP:conf/wine/GiannakopoulosZ18} with $c = - \frac{\log(q)}{H_n - H_{j-1}}$ and apply the quantile function on both sides proves the result. 
\end{proof}

Now, we are ready to prove the theorem.

\begin{proof}[Proof of Theorem~\ref{theorem:alpha_correlation_dynamic-positive}]
	First, observe that in principle, buyers will be indifferent between two items $j$ and $j'$ if $\alpha_j = \alpha_{j'}$. As these items are indistinguishable for later buyers anyway and new prices will be defined, we can assume that ties are broken in our favor. That is why we can assume in the following that a buyer $i$ will prefer item $\ell_k$ if and only if $F^{-1}\left( 1 - \frac{k}{n - i + 1}\right) \leq v_i < F^{-1}\left( 1 - \frac{k-1}{n - i + 1}\right)$. Hence, the social welfare achieved by the matching computed via Algorithm \ref{algorithm_dynamic_correlated} is equivalent to the social welfare of the posted-pricing mechanism. 
	
	\begin{algorithm}[h]
		\SetAlgoNoLine
		\DontPrintSemicolon
		$M^{(1)} \longleftarrow M$ \\
		\For{$i \in N$}{
			\text{Let $k$ be such that $F^{-1}\left( 1 - \frac{k}{n - i + 1}\right) \leq v_i < F^{-1}\left( 1 - \frac{k-1}{n - i + 1}\right)$} \\
			\text{Set $j_i = \ell_k$ (i.e. match $i$ and $\ell_k$) and remove $\ell_k$ from $M^{(i+1)}$ } 
		}
		\caption{Quantile Allocation Rule}
		\label{algorithm_dynamic_correlated}
	\end{algorithm}
	
	For any buyer $i$ and available item $j \in M^{(i)}$, let $w^{(i)}_j$ be smallest type for buyer $i$ such that item $j$ is the preferred one. Furthermore, let $X_{i, j} = 1$ if item $j$ gets allocated in step $i$.
	
	Observe that the social welfare of the final allocation is at least
	\[
	\sum_{i = 1}^n \sum_{j = 1}^n \alpha_j w^{(i)}_j X_{i, j} \enspace.
	\]
	Furthermore, by our construction
	\[
	\Pr{X_{i, j} = 1} = \frac{1}{n - i + 1} \prod_{i' = 1}^{i-1} \left( 1 - \frac{1}{n - i' + 1} \right) = \frac{1}{n} \enspace.
	\]
	Note that also $w^{(i)}_j$ is a random variable and our goal will be to lower-bound its expectation. Let $Z = j - 1 - \sum_{i' = 1}^{i - 1} \sum_{j' = 1}^{j - 1} X_{i', j'}$ be a random variable indicating the number of items of smaller index that are included in $M^{(i)}$. By this definition $w^{(i)}_j = F^{-1}\left( 1 - \frac{Z + 1}{n - i + 1} \right)$.
	
	We distinguish two cases. First, consider $j \leq n^{2/3}$. We use that $Z + 1 \leq j$. Therefore, by Lemma~\ref{lemma:quantiles1}, we have
	\[
	w^{(i)}_j = F^{-1}\left( 1 - \frac{Z + 1}{n - i + 1} \right) \geq F^{-1}\left( 1 - \frac{j}{n - i + 1} \right) \geq  \frac{\log(n - i + 1) - \log j}{H_n - H_{j-1}} \mu_j \enspace.
	\]
	Observe that this lower bound is not a random variable. Therefore, we can write
	\[
	\Ex{\sum_{i = 1}^n w^{(i)}_j X_{i, j}} \geq \sum_{i = 1}^n \frac{\log(n - i + 1) - \log j}{H_n - H_{j-1}} \mu_j \Ex{X_{i, j}} \enspace.
	\]
	By applying the integral estimation $\sum_{i = 1}^n \log(n - i + 1) \geq n \log n - n + 1$ and using that $\Ex{X_{i, j}} = \frac{1}{n}$, this term is at least
	\[
	\frac{\log n - 1 + \frac{1}{n} - \log j}{H_n - H_{j-1}} \mu_j \geq \frac{\log n - \log j - 1}{\log n - \log j + 2} \mu_j = \left( 1 - \frac{3}{\log n - \log j + 2} \right) \mu_j \geq \left( 1 - \frac{9}{\log n} \right) \mu_j \enspace.
	\]
	
	For $j > n^{\frac{2}{3}}$, we only consider $i \leq \left( 1 - \frac{1}{\log n} \right) n$. Note that $Z$ is essentially the outcome of sampling without replacement. Conditional on $j \in M^{(i)}$, we have
	\[
	\Pr{j' \in M^{(i)} \growingmid j \in M^{(i)}} = \frac{n - i}{n - 1}
	\]
	and so
	\[
	\Ex{Z \growingmid j \in M^{(i)}} = (j-1) \frac{n - i}{n - 1} \enspace.
	\]
	We can apply Hoeffding's inequality and get
	\[
	\Pr{Z \leq (j-1) \frac{n - i}{n - 1} - \sqrt{n \log n} \growingmid j \in M^{(i)}} \leq \exp\left( - 2 \frac{n \log n}{i - 1} \right) \leq \frac{1}{n^2} \enspace.
	\]
	
	Observe that $(j - 1)(n - i) + (n - 1) \geq j (n - i + 1) - n$ so, if $Z > (j-1) \frac{n - i}{n - 1} - \sqrt{n \log n}$, then also $Z + 1 > \frac{j}{n-1}(n - i + 1) - \frac{n}{n-1} - \sqrt{n \log n}$. Therefore, if $i \leq \left( 1 - \frac{1}{\log n} \right) n$
	\[
	\Pr{\frac{Z + 1}{n - i + 1} > \frac{j}{n} - \frac{2 (\log n)^{3/2}}{\sqrt{n}} \growingmid j \in M^{(i)}} \geq \Pr{\frac{Z + 1}{n - i + 1} > \frac{j}{n} - \frac{2 + \sqrt{n \log n}}{n - i + 1} \growingmid j \in M^{(i)}} \geq 1 - \frac{1}{n^2} \enspace.
	\]
	So, by Lemma~\ref{lemma:quantiles2}, we get that
	\[
	\Pr{w^{(i)}_j > \mu_{\lfloor j + \sqrt{n \log n} + 2 \sqrt{n} (\log n)^{3/2} \rfloor}  \growingmid j \in M^{(i)}} \geq 1 - \frac{1}{n^2} \enspace.
	\]
	Consequently, 
	\[
	\Ex{w^{(i)}_j X_{i, j}} \geq \left( 1 - \frac{1}{n^2} \right) \cdot \mu_{\lfloor j + \sqrt{n \log n} + 2 \sqrt{n} (\log n)^{3/2} \rfloor} \Ex{X_{i, j}} \enspace.
	\]
	Therefore, we get
	\[
	\Ex{\sum_{i = 1}^n w^{(i)}_j X_{i, j}} \geq \Ex{\sum_{i = 1}^{(1 - \frac{1}{\log n})n} w^{(i)}_j X_{i, j}} \geq \left(1 - \frac{1}{\log n} \right) \mu_{\lfloor j + \sqrt{n \log n} + 2 \sqrt{n} (\log n)^{3/2} \rfloor} \enspace.
	\]
	
	Combining these bounds for all $j$, we get
	\begin{align*}
	\Ex{\sum_{i = 1}^n \sum_{j = 1}^m \alpha_j w^{(i)}_j X_{i, j}} & = \sum_{j = 1}^{\lfloor n^{2/3} \rfloor} \Ex{\sum_{i = 1}^n \alpha_j w^{(i)}_j X_{i, j}} + \sum_{j = \lfloor n^{2/3} \rfloor + 1}^{n} \Ex{\sum_{i = 1}^n \alpha_j w^{(i)}_j X_{i, j}} \\
	& \geq \left(1 - \frac{9}{\log n} \right) \sum_{j = 1}^{\lfloor n^{2/3} \rfloor} \alpha_j \mu_j + \left( 1 - \frac{1}{\log n} \right) \sum_{j = \lfloor n^{2/3} \rfloor + \lfloor \sqrt{n \log n} + 2 \sqrt{n} (\log n)^{3/2} \rfloor + 1}^{n} \alpha_j \mu_j \\
	& \geq \left( 1 - O\left( \frac{1}{\log n} \right) \right) \sum_{j = 1}^m \alpha_j \mu_j \enspace. \qedhere
	\end{align*}
\end{proof}

	\section{Proofs for Section \ref{section:alpha_correlation_static} on Separable Valuations}
\label{appendix:correlated_valuations_static}

We give a full proof of Theorem \ref{theorem:alpha_correlation_staticpositive}.

\alphacorrelationstaticpositive*

Before proving the theorem, we need two lemmas. First, we argue that all items are allocated with a reasonable high probability. Second, we give a bound on the quantiles of MHR distributions. \\

To begin with the first lemma, note that the prices used by the mechanism are fairly low so that we can guarantee that all items $1,\dots,m=n-n^{5/6}$ are sold with reasonably high probability.

\begin{lemma}
	With probability at least $1 - \frac{111}{\log n}$ all items $1, \ldots, m$ are allocated.
\end{lemma}

\begin{proof}
	Let $M^\ast = \{j \mid \alpha'_j > \alpha'_{j+1}\}$ be the set of all items which are strictly better than the following one. Furthermore, let $Z_j$ be the number of buyers of value more than $F^{-1}(1 - q_j)$. If $j \in M^\ast$ and $v_i > F^{-1}(1 - q_j)$, then for all $j' > j$
	\[
	\alpha'_j v_i - p_j = \alpha'_j v_i - \left( \sum_{k = j}^{j'-1} (\alpha'_k - \alpha'_{k+1}) F^{-1}(1 - q_k) + p_{j'} \right) > \alpha'_j v_i - \sum_{k = j}^{j'-1} (\alpha'_k - \alpha'_{k+1}) v_i - p_{j'} = \alpha'_{j'} v_i - p_{j'} \enspace.
	\]
	That is, any buyer of value more than $F^{-1}(1 - q_j)$ prefers item $j$ over any item $j+1, \ldots, m$, albeit item $j$ might not be their first choice.
	
	Consider any $j \in M^\ast$. Let $j''$ be the smallest index such that $\alpha'_{j''} = \alpha'_j$ (possibly $j'' = j$). We claim that in the event $Z_j \geq j$, items $\{j'', \ldots, j\}$ are allocated for sure. To see this, we observe that at most $j'' - 1$ buyers can buy one of the items $1, \ldots, j'' - 1$. That is, there are at least $Z_j - (j'' - 1) \geq j - j'' + 1$ buyers who do not buy one of the items $1, \ldots, j'' - 1$ but have value more that $F^{-1}(1 - q_j)$, i.e. they prefer every item $j'', \ldots, j$ to any item $j' > j$. Consequently, each of them buys one of the items $j'', \ldots, j$ while they are still available.
	
	Due to this observation, all that remains is to show that with probability at least $1 - \frac{111}{\log n}$ we have $Z_j \geq j$ for all $j \in M^\ast$. To this end, we distinguish between two ranges for $j$ and let $j^\ast = \frac{\sqrt{n \log n}}{2 \log \log n - 1}$. \\
	
	For $j \leq j^\ast$, we have $q_j = \frac{j}{n} 2 \log \log n$ and therefore $\Ex{Z_j} = 2 j \log \log n$. A Chernoff bound yields for the probability that $Z_j$ is less than $j$
	\begin{align*}
	\Pr{Z_j < j} &= \Pr{Z_j < (1 - \delta) \Ex{Z_j}} \leq \exp\left( - \frac{1}{2} \delta^2 \Ex{Z_j} \right) \\[0.5em]
	&= \exp\left( - \frac{(\log \log n - 2)^2}{\log \log n} j \right) \leq \exp\left( - \left( \log \log n - 4 \right) j \right) = \left(\frac{\e^4}{\log n}\right)^j \enspace.
	\end{align*}
	
	For $j > j^\ast$, we have $q_j = \frac{j}{n} + \sqrt{\frac{\log n}{n}}$ and so $\Ex{Z_j} = j + \sqrt{n \log n}$. Therefore, Hoeffding's inequality gives us
	\[
	\Pr{Z_j < j} = \Pr{Z_j < \Ex{Z_j} - \sqrt{n \log n}} \leq \exp\left( - 2 \frac{(\sqrt{n \log n})^2}{n} \right) = \frac{1}{n^2} \enspace.
	\]
	
	In combination, applying the union bound gives us
	\[
	\Pr{\exists j \in M^\ast: Z_j < j} \leq \sum_{j = 1}^{j^\ast} \left(\frac{\e^4}{\log n}\right)^j + n \frac{1}{n^2} \leq 2 \frac{\e^4}{\log n} + \frac{1}{n} \leq \frac{111}{\log n} \enspace. \qedhere
	\]
\end{proof}

Now, we are ready to prove Theorem~\ref{theorem:alpha_correlation_staticpositive}.

\begin{proof}[Proof of Theorem~\ref{theorem:alpha_correlation_staticpositive}]
	Our point of comparison will be the VCG mechanism restricted to the first $m$ items. Its expected welfare is
	\[
	\Ex{\sum_{j = 1}^n \alpha'_j v_{(j)}} \geq \frac{m}{n} \Ex{\sum_{j = 1}^n \alpha_j v_{(j)}} \geq \left( 1 - \frac{1}{n^{1/6}}\right) \Ex{\sum_{j = 1}^n \alpha_j v_{(j)}} \enspace.
	\]
	
	To compare the welfare of our mechanism to this VCG welfare, we split it up into the revenue collected by either mechanism and the sum of buyers' utilities. 
	
	\paragraph*{Revenue comparison. } With probability $1 - \frac{111}{\log n}$, all items $1, \ldots, m$ are allocated by our mechanism. Therefore, we can bound the expected revenue
	\[
	\Ex{\sum_{\text{$j$ is sold}} p_j} \geq \left(1 - \frac{111}{\log n}\right) \sum_{j = 1}^m p_j = \left(1 - \frac{111}{\log n}\right) \sum_{j = 1}^m \sum_{k = j}^m (\alpha'_k - \alpha'_{k+1}) F^{-1}(1 - q_k) \enspace.
	\]
	The expected revenue of VCG is
	\[
	\Ex{\sum_{j = 1}^m \sum_{k = j}^m (\alpha'_k - \alpha'_{k+1}) v_{(k+1)}} = \sum_{j = 1}^n \sum_{k = j}^n (\alpha'_k - \alpha'_{k+1}) \mu_{k + 1} \enspace.
	\]
	
	For $k \leq n^{5/6}$, we use Lemma~\ref{lemma:quantiles1} to get
	\begin{align*}
	F^{-1}(1 - q_k) & \geq - \frac{\log(q_k)}{H_n - H_{k-1}} \mu_k \geq \frac{\log n - \log(k 2 \log \log n)}{H_n - H_{k-1}} \mu_k \geq \frac{\log n - \log(k 2 \log \log n)}{\log n - \log k + 2} \mu_k \\
	& \geq \left( 1 - \frac{2 + \log 2 + \log \log \log n}{\frac{1}{6} \log n + 2} \right) \mu_k \geq \left( 1 - 20 \frac{\log \log \log n}{\log n} \right) \mu_k \enspace.
	\end{align*}
	
	For $k > n^{5/6}$, we also use Lemma~\ref{lemma:quantiles1}
	\begin{align*}
	F^{-1}(1 - q_k) & \geq - \frac{\log(q_k)}{H_n - H_{k-1}} \mu_k \geq \frac{- \log\left(\frac{k}{n} + \sqrt{\frac{\log n}{n}} \right)}{H_n - H_{k-1}} \mu_k \geq \frac{\log n - \log(k + \sqrt{n \log n})}{\log n - \log(k-1)} \mu_k
	\end{align*}
	By concavity of the $\log$ function, we have $\log(k + \sqrt{n \log n}) \leq \log(k-1) + \frac{1 + \sqrt{n \log n}}{k - 1}$. So
	\[
	\frac{\log n - \log(k + \sqrt{n \log n})}{\log n - \log (k - 1)} \mu_k \geq \frac{\log n - \log(k-1) - \frac{1 + \sqrt{n \log n}}{k - 1}}{\log n - \log(k-1)} \mu_k = \left( 1 - \frac{\frac{1 + \sqrt{n \log n}}{k - 1}}{\log n - \log(k-1)} \right) \mu_k \enspace.
	\]
	Using $n^{5/6} \leq k - 1 \leq n - n^{5/6}$, we have
	\[
	(k - 1) \log\left( \frac{n}{k - 1} \right) \geq n^{5/6} \log \frac{1}{1 - n^{-1/6}} = n^{2/3} \log \left( \frac{1}{1 - n^{-1/6}} \right)^{n^{1/6}} \geq n^{2/3}
	\]
	and therefore, if $n$ is large enough,
	\[
	\left( 1 - \frac{\frac{1 + \sqrt{n \log n}}{k - 1}}{\log n - \log(k-1)} \right) \mu_k \geq \left( 1 - \frac{1 + \sqrt{n \log n}}{n^{2/3}} \right) \mu_k \geq \left( 1 - 20 \frac{\log \log \log n}{\log n} \right) \mu_k \enspace.
	\]
	
	By comparing the coefficients, we observe that the revenue of the static price mechanism is within a $1 - O\left(\frac{\log \log \log n}{\log n}\right)$-factor of the VCG revenue.
	
	\paragraph*{Utility comparison. } To compare the buyers' utilities, we exploit that both VCG as well as our static-price mechanism are truthful. Recall that the \emph{virtual value} associated to value $t$ is given by $\phi(t) = t - \frac{1}{h(t)}$. By Myerson's theory \citep{DBLP:journals/mor/Myerson81}, we know that for any truthful mechanism, the expected revenue is equal to the expected virtual welfare. For any truthful mechanism, letting $\textsc{sold}$ denote the set of items that are allocated and, for $j \in \textsc{sold}$, letting $b_j$ be the value of the buyer being allocated the $j$ item, this allows us to rewrite the expected sum of utilities as
	\[
	\Ex{\sum_{j \in \textsc{sold}} \alpha_j b_j} - \Ex{\sum_{j \in \textsc{sold}} \alpha_j \phi(b_j)} = \Ex{\sum_{j \in \textsc{sold}} \alpha_j \frac{1}{h(b_j)}} \enspace.
	\]
	
	We now fix a valuation profile $v$ and compare $\sum_{j \in \textsc{sold}} \alpha_j \frac{1}{h(b_j)}$ between the VCG and the static-price mechanism. By superscripts, we distinguish between the VCG outcome and the static-price outcome (sp). For VCG, we have $b^{\mathrm{VCG}}_j = v_{(j)}$ and $\textsc{sold}^{\mathrm{VCG}} = \{1, \ldots, m\}$. If $\textsc{sold}^{\mathrm{sp}} \supseteq \{1, \ldots, m\}$, we claim that
	\[
	\sum_{j \in \textsc{sold}^{\mathrm{VCG}}} \alpha'_j \frac{1}{h(b^{\mathrm{VCG}}_j)} \leq \sum_{j \in \textsc{sold}^{\mathrm{sp}}} \alpha'_j \frac{1}{h(b^{\mathrm{sp}}_j)} \enspace.
	\]
	This is because for all $\ell$ we have
	\[
	\sum_{j = 1}^\ell \frac{1}{h(b^{\mathrm{VCG}}_j)} = \sum_{j = 1}^\ell \frac{1}{h(v_{(j)})} \leq \sum_{j = 1}^\ell \frac{1}{h(b^{\mathrm{sp}}_j)}
	\]
	due to the fact that $v_{(1)}, \ldots, v_{(\ell)}$ are the largest $\ell$ entries in $v$ while $b^{\mathrm{sp}}_1, \ldots, b^{\mathrm{sp}}_\ell$ are any $\ell$ entries in $v$ and $\frac{1}{h(\cdot)}$ is non-increasing.
	This then combines to
	\[
	\sum_{j \in \textsc{sold}^{\mathrm{VCG}}} \alpha'_j \frac{1}{h(b^{\mathrm{VCG}}_j)} = \sum_{\ell = 1}^m (\alpha'_\ell - \alpha'_{\ell + 1}) \sum_{j = 1}^\ell \frac{1}{h(b^{\mathrm{VCG}}_j)} \leq \sum_{\ell = 1}^m (\alpha'_\ell - \alpha'_{\ell + 1}) \sum_{j = 1}^\ell \frac{1}{h(b^{\mathrm{sp}}_j)} \leq \sum_{j \in \textsc{sold}^{\mathrm{sp}}} \alpha'_j \frac{1}{h(b^{\mathrm{sp}}_j)} \enspace.
	\]
	
	If $\textsc{sold}^{\mathrm{sp}} \nsupseteq \{1, \ldots, m\}$, we simply use the fact that the sum of utilities is non-negative. Therefore, we get
	\[
	\Ex{\sum_{j \in \textsc{sold}^{\mathrm{sp}}} \alpha'_j \frac{1}{h(b^{\mathrm{sp}}_j)}} \geq \left( 1 - \frac{111}{\log n} \right) \Ex{\sum_{j \in \textsc{sold}^{\mathrm{VCG}}} \alpha'_j \frac{1}{h(b^{\mathrm{VCG}}_j)}} \enspace.
	\]
\end{proof}

	\section{Proofs for Section \ref{section_optimality_dynamic} on Upper Bounds}
\label{appendix:optimality_dynamic}

We give a proof of Lemma \ref{lemma:mdp-recursion}.

\mdprecursion*

\begin{proof}
	We are going to prove the statement by induction on index variable $k$.\\
	
	First, consider the induction base $k = 2$. By definition of the thresholds we know that $p^{(n)} = 0$ and $p^{(n-1)} = \Ex{\max \{v_n, p^{(n)} \}} = \Ex{\max \{v_n, 0 \}} = \Ex{v_n} = 1$. Next, consider the threshold $p^{(n-2)}$ defined by
	\begin{align*}
		p^{(n-2)} &= \Ex{\max \{v_{n-1}, p^{(n-1)} \}} = \Ex{\max \{v_{n-1}, 1 \}} \\
		&= \int_{0}^{\infty} \Pr{\max\{v_{n-1}, 1\} \geq x} \mathrm{d}x = \int_{0}^{1} 1 \mathrm{d}x + \int_{1}^{\infty} \e^{-x} \mathrm{d}x = 1 + \frac{1}{\e} \approx 1.368 \,.
	\end{align*}
	Thus, we can easily verify that $p^{(n-2)} \leq 1.375 = H_2 - \frac{1}{8}$.\\
	
	For the inductive step, we move from $k$ to $k+1$. By the induction hypothesis, we have
	\[
	p^{(n-(k+1))} = \Ex{\max \{v_{n-k}, p^{(n-k)} \}} \leq \Ex{\max \left\{v_{n-k}, H_k - \frac{1}{8} \right\}} \enspace.
	\]
	Furthermore
	\begin{align*}
		\Ex{\max \left\{v_{n-k}, H_k - \frac{1}{8} \right\}} & = \int_{0}^{\infty} \Pr{\max\left\{v_{n-k}, H_k - \frac{1}{8}\right\} \geq x} \mathrm{d}x = \int_{0}^{H_k - \frac{1}{8}} 1 \mathrm{d}x + \int_{H_k - \frac{1}{8}}^{\infty} \e^{-x} \mathrm{d}x \\
		& = H_k - \frac{1}{8} + \e^{-H_k + \frac{1}{8}} \enspace.
	\end{align*}
	We now use the fact that the $k$-th harmonic number $H_k$ for $k \geq 2$ is bounded from below by $H_k \geq \log k + \gamma$, in which $\gamma \approx 0.577$ denotes the Euler-Mascheroni constant. So for $k \geq 2$
	\[
	\e^{-H_k + \frac{1}{8}} \leq \e^{-(\log k + \gamma) + \frac{1}{8}} = \frac{\e^{\frac{1}{8} - \gamma}}{k} \leq \frac{\e^{\frac{1}{8} - 0.57}}{k} \leq \frac{1}{\frac{3}{2} \cdot k} \leq \frac{1}{k+1} \enspace.
	\]
	In combination, this gives us
	\[
	p^{(n-(k+1))} \leq H_k - \frac{1}{8} + \frac{1}{k+1} = H_{k+1} - \frac{1}{8} \enspace.
	\]
\end{proof}

	\section{Proofs for Section \ref{section_optimality_static} on Upper Bounds}
\label{appendix:optimality_static}

We give a full proof of Proposition \ref{proposition:optimality-static}.

\optimalitystatic*

\begin{proof}[Proof of Proposition \ref{proposition:optimality-static}]
	Consider $n \geq \e^{\e^{\e^4}}$.
	
	Observe that always $\Ex{\OPT} = \Ex[]{\max_{i \in [n]} v_i} = H_n$ \citep{Arnold:2008:FCO:1373327}. Now, we will bound $\Ex{\ALG}$ for any choice of a static price $p \in \mathbb{R}_{\geq 0}$. Regardless of $p$, we have
	\[
	\Ex{\ALG} = \Ex[]{v \growingmid v \geq p} \cdot \Pr{\exists i: v_i \geq p} = (p + 1) \cdot \left(1 - \left( 1 - \e^{-p} \right)^n \right) \enspace.
	\]
	
	We will show that for any choice of a static price $p$,
	\[
	(p + 1) \cdot \left(1 - \left( 1 - \e^{-p} \right)^n \right) \leq H_n - c \log \log \log n
	\]
	for some constant $c$, which then immediately proves the claim as $H_n = \Theta(\log n)$.
	
	To this end, we will consider three cases for the choice of $p$.
	
	\paragraph*{Case 1: } $0 \leq p < \log n - \frac{1}{2} \log \log \log n$: We use the trivial upper bound of $1$ for the probability term in $\Ex{\ALG}$, so
	\begin{align*}
	(p + 1) \cdot \left(1 - \left( 1 - \e^{-p} \right)^n \right) & \leq (p + 1)  < \log n - \frac{1}{2} \log \log \log n + 1 \\
	& \leq \log n - \frac{1}{2} \log \log \log n + \frac{1}{4} \log \log \log n = \log n - \frac{1}{4} \log \log \log n
	\end{align*}
	as $n \geq \e^{\e^{\e^4}}$ holds.
	
	\paragraph*{Case 2: } $\log n - \frac{1}{2} \log \log \log n \leq p \leq H_n - 1$: Again, observe that the expected value of the algorithm can be upper bounded by
	\begin{align*}
	&(p + 1) \cdot \left(1 - \left( 1 - \e^{-p} \right)^n \right) \leq H_n \cdot \left(1 - \left( 1 - \e^{-p} \right)^n \right) \\
	&\leq H_n \cdot \left(1 - \left( 1 - \e^{-\log n + \frac{1}{2} \log \log \log n} \right)^n \right) = H_n \cdot \left(1 - \left( 1 - \frac{\sqrt{\log \log n}}{n} \right)^n \right) \,.
	\end{align*}
	
	Next, we want to lower bound $\left( 1 - \frac{\sqrt{\log \log n}}{n} \right)^n$ in order to get the desired upper bound on $\Ex[]{\ALG}$. For this purpose, we use the following inequality. 
	For $n > 1$ and $x \in \mathbb{R}$ with $\lvert x\rvert \leq n$, we have $\left( 1 + \frac{x}{n} \right)^n \geq \e^x \cdot \left( 1 - \frac{x^2}{n} \right)$. This way, we get
	\begin{align*}
	\left( 1 - \frac{\sqrt{\log \log n}}{n} \right)^n &\geq \e^{-\sqrt{\log \log n}} \cdot \left( 1 - \frac{(- \sqrt{\log \log n})^2}{n} \right) = \e^{-\sqrt{\log \log n}} \cdot \underbrace{\left( 1 - \frac{\log \log n}{n} \right)}_{\geq \frac{1}{2} ,\;\forall n} \geq \frac{1}{2} \e^{-\sqrt{\log \log n}} \,.
	\end{align*}
	
	Note that if $\log \log n \geq 4$, then also $\sqrt{\log \log n} \leq \frac{1}{2} \log \log n$, so $\e^{-\sqrt{\log \log n}} \geq \e^{- \frac{1}{2} \log \log n} = \frac{1}{\sqrt{\log n}}$. This gives us
	\begin{align*}
	(p + 1) \cdot \left(1 - \left( 1 - \e^{-p} \right)^n \right) &\leq H_n \cdot \left(1 - \left( 1 - \frac{\sqrt{\log \log n}}{n} \right)^n \right) \leq H_n \cdot \left( 1 - \frac{1}{2} \e^{-\sqrt{\log \log n}} \right) \\
	&\leq H_n \cdot \left( 1 - \frac{1}{2\sqrt{\log n}} \right) \leq H_n \cdot \left( 1 - \frac{\log \log \log n}{2 \log n} \right) \enspace,
	\end{align*}
	where in the last step we use that $\sqrt{\log n} \geq \log \log \log n$ and therefore $\sqrt{\log n} \leq \frac{\log n}{\log \log \log n}$.
	
	\paragraph*{Case 3: } $p > H_n - 1$: In this case we use the fact that \citet{DBLP:conf/wine/GiannakopoulosZ18} showed that the revenue function $p \mapsto p \cdot \left(1 - \left( 1 - \e^{-p} \right)^n\right)$ is non-increasing on $[H_n - 1, \infty)$. 
	
	We can conclude that
	\begin{align*}
	(p + 1) \cdot \left(1 - \left( 1 - \e^{-p} \right)^n \right) & \leq p \cdot \left(1 - \left( 1 - \e^{-p} \right)^n \right)+ 1 \leq H_n \cdot \left(1 - \left( 1 - \e^{-(H_n - 1)} \right)^n \right)+ 1 \\
	&\leq H_n \cdot \left(1 - \left( 1 - \frac{\e}{n} \right)^n \right)+ 1 \leq \frac{99}{100} H_n + 1
	\end{align*}
\end{proof}

	% !TEX root = mhr-prophet.tex
\section{Dynamic Pricing for Subadditive Valuations with MHR marginals}
\label{appendix:mhr_marginals_dynamic}

Concerning subadditivity, we first start with a straight-forward extension of the definitions from Section \ref{section:preliminiaries} to the case of buyers' valuation functions being subadditive. The experienced reader can skip these definitions and go to the paragraph on dynamic-pricing directly.

\subsubsection*{Notation and Definitions for Subadditive Valuation Functions}

As before, we consider a setting of $n$ buyers and a set $M$ of $m$ items. In contrast, now every buyer has a valuation function $v_i\colon 2^M \to \mathbb{R}_{\geq 0}$ mapping each bundle of items to the buyer's valuation which is subadditive, i.e. $v_i(S \cup T) \leq v_i(S) + v_i(T)$ for any $S, T \subseteq M$ for all $v_i$. The functions $v_1, \ldots, v_n$ are drawn independently from a publicly known distribution $\mathcal{D}$ which has \emph{MHR marginals}. We define MHR marginals as follows. Let $\mathcal{D}_j$ be the marginal distribution of $v_i(\{j\})$, which is the value of a buyer for being allocated only item $j$. We assume that $\mathcal{D}_j$ is a continuous, real, non-negative distribution with monotone hazard rate. Note that this allows arbitrary correlation between the items.

Buyer $i$ picks the bundle of items $S \subseteq M^{(i)}$ which maximizes her \emph{utility} $v_i(S) - \sum_{j \in S} p_j^{(i)}$ if positive. The definitions of social welfare and revenue are the natural extensions from Section \ref{section:preliminiaries}, that is let $S_i$ denote the (possibly empty) bundle of items allocated to buyer $i$ by our mechanism. The \emph{expected social welfare} of the mechanism is given by $\Ex[]{\sum_{i = 1}^n v_{i} (S_i) } =: \Ex[]{\ALG}$. Its \emph{expected revenue} is given by $\Ex[]{\sum_{i = 1}^n \sum_{j \in S_i} p^{(i)}_{j}} =: \Ex[]{\revpp}$. In comparison, let the social welfare maximizing allocation assign bundle $S_i^*$ to buyer $i$. Its expected social welfare is therefore given by $\Ex[]{\sum_{i = 1}^n v_{i} (S_i^*) } =: \Ex[]{\OPT}$. Using the subadditivity of buyers, we can upper-bound the expected optimal social welfare by $\Ex[]{\sum_{i = 1}^n v_{i} (S_i^*) } \leq \Ex[]{\sum_{i = 1}^n \sum_{j \in S_i^*}v_{i} (\{j\}) } \leq \sum_{j=1}^{m} \Ex[]{\max_{i \in [n]} v_i(\{j\}) }$. Furthermore, let $\Ex{\revopt}$ denote the maximum expected revenue of any individually rational mechanism. Due to individual rationality, we have $\Ex{\revopt} \leq \Ex{\OPT}$ and $\Ex{\revpp} \leq \Ex{\ALG}$.

We can achieve asymptotically optimal bounds on the expected welfare and revenue by assuming subadditive buyers as well as MHR distributions on the valuations for single items $v_i(\{j\})$. This way, one can naturally extend the well-known MHR definition for the value of a single item to the case of multiple heterogeneous items with possible correlations among them. Still, our bounds only give asymptotically full efficiency if the number of items satisfies $m = n^{o(1)}$. Pushing this result further is a desirable goal for future research. 

\subsubsection*{Dynamic Prices}
\label{section:mhr_marginals_dynamic}
We first consider dynamic prices. That is, buyer $i$ faces prices depending on the set of available items. Our strategy is to sell specific items only to a subgroup of buyers in order to gain control over the selling process. We can implement this by imposing the following item prices which decrease as the selling process proceeds.

We split the set of buyers in groups of size $\left\lfloor \frac{n}{m} \right\rfloor =: n'$ and sell item $j$ only to the group of buyers \[ \left\{(j-1)\cdot n' + 1,..., (j-1)\cdot n' + n' \right\} =: N_j \enspace . \] For the $k$-th buyer in $N_j$, we set the price for item $j$ to 
\begin{align} \label{prices_dynamic}
p_j^{((j-1) n' + k)} = F_j^{-1}\left( 1 - \frac{1}{\left\lfloor \frac{n}{m} \right\rfloor - k + 1 } \right)
\end{align} and the prices for all other unsold items to infinity. This choice of prices ensures that the first item is sold among the first $\left\lfloor \frac{n}{m} \right\rfloor$ buyers, the second item among the second $\left\lfloor \frac{n}{m} \right\rfloor$ buyers etc. As a consequence, all items are sold in our process. 

\begin{restatable}{theorem}{dynamicmhrmarginals}
	\label{theorem:dynamic_mhr_marginals}
	The posted-prices mechanism with subadditive buyers and dynamic prices is \newline $1 - O \left( \frac{1 + \log m}{\log n} \right)$-competitive with respect to social welfare. 
\end{restatable}

\begin{proof}
	We start by considering the case of selling one item among $n' = \left\lfloor \frac{n}{m} \right\rfloor$ buyers where the prices for the item are as in equation \ref{prices_dynamic}. Note that we simplify notation in this context and omit the index of the item. Let $X_i$ be a random variable which is equal to $1$ if buyer i buys the item and $0$ otherwise. Note that by our choice of prices, the item is sold in step $i$  with probability $\frac{1}{n' -i +1}$ what leads to \[ \Ex[]{X_i} = \Pr{X_i = 1} = \frac{1}{n' - i + 1} \cdot \prod_{i' = 1}^{i-1} \left( 1 - \frac{1}{n' - i' + 1} \right) = \frac{1}{n'} \ .\]
	
	Further, buyer $i$ only buys the item if $v_{i}$ exceeds the price. Using Lemma \ref{lemma:quantiles1} allows to calculate 
	\begin{align*}
	p^{(i)} = F^{-1}\left( 1 - \frac{1}{n' - i + 1 } \right) \geq \frac{\log \left( n' - i + 1 \right)}{H_{n'}} \Ex[]{\max_{i \in [n']} v_i} \ . 
	\end{align*}
	Note that this bound is deterministic. An application of the integral estimation $\sum_{i=1}^{n'} \log \left( n' -i+1 \right) \geq n' \log n' - n' + 1$ as well as bounding the harmonic number $H_{n'} \leq \log n' +1$ lead to a lower bound for the expected social welfare: 
	\begin{align*}
	\sum_{i = 1}^{n'} \Ex[]{p^{(i)} X_i} & \geq \sum_{i = 1}^{n'} \frac{\log \left( n' - i + 1  \right)}{n' H_{n'}} \Ex[]{\max_{i \in [n']} v_i} \geq \frac{n' \log n' - n' + 1}{n' H_{n'}} \Ex[]{\max_{i \in [n']} v_i} \\ &\geq \frac{\log n' - 1 + \frac{1}{n'}}{\log n' + 1} \Ex[]{\max_{i \in [n']} v_i} \geq \left( 1 - \frac{2}{\log n'} \right) \Ex[]{\max_{i \in [n']} v_i}
	\end{align*} 
	
	Now, we apply Lemma \ref{Lemma:Babaioff_bound_maximum} which states that the quotient of the expectation of the maximum of $n'$ and $n$ i.i.d. random variables is lower bounded by $\log n' / \log n$ for $n' \leq n$. This leads to the bound \[ \Ex[]{\max_{i \in [n']} v_i} \geq \frac{\log n'}{\log n} \Ex[]{\max_{i \in [n]} v_i} = \left( 1 - O \left( \frac{1 + \log m}{\log n} \right) \right) \Ex[]{\max_{i \in [n]} v_i} \ . \]
	
	In order to generalize this to the case of $m$ items, our pricing strategy ensures that we can apply the received bound for every item separately. To this end, note that only buyers in $N_j$ will consider buying item $j$. Further, also by our prices, every buyer will buy at most one item. By the introduction of indicator random variables $X_{ij}$ indicating if buyer $i$ buys item $j$, we can conclude 
	\begin{align*}
	\Ex[]{\revpp} & = \sum_{j=1}^{m} \sum_{i=1}^{n} \Ex[]{p^{(i)}_j X_{i, j}} = \sum_{j=1}^{m} \sum_{i \in N_j} \Ex[]{p^{(i)}_j X_{i, j}} \\ & \geq \left( 1 - \frac{2}{\log n'} \right) \sum_{j=1}^{m} \Ex[]{\max_{i \in [n']} v_{i}(\{j\}) } \\ & \geq \left( 1 - O \left( \frac{1 + \log m}{\log n} \right) \right) \sum_{j=1}^{m} \Ex[]{\max_{i \in [n]} v_{i}(\{j\})} \\ & \geq \left( 1 - O \left( \frac{1 + \log m}{\log n} \right) \right) \Ex[]{\OPT}  \ .
	\end{align*}
\end{proof}

\begin{corollary}
	The expected revenue of the posted-prices mechanism with subadditive buyers and dynamic prices is a $1 - O \left( \frac{1 + \log m}{\log n}\right)$-fraction of the expected optimal revenue. 
\end{corollary} 

Note that the assumption of buyers' valuations being identically distributed is actually a too strong requirement for these results. For the proofs in this chapter it would be sufficient to consider buyers having identical marginals on single item sets, but correlations between items might be buyer-specific.

	% !TEX root = mhr-prophet.tex
\section{Static Pricing for Subadditive Valuations with MHR marginals}
\label{appendix:mhr_marginals_static}

The notation for subadditive valuation functions as well as the definition of MHR marginals are given in Appendix \ref{appendix:mhr_marginals_dynamic}. For the case of static prices, we give a $1 - O \left( \frac{\log\log\log n}{\log n} + \frac{\log m}{\log n}\right)$-competitive mechanism. 

The general design idea for our mechanism is as follows. Setting fairly low prices will put high probability on the event of selling all items. Although we cannot control which buyer will buy which bundle of items, we can extract all social welfare of the posted prices mechanism as revenue. Therefore, having prices which still ensure that the revenue can be lower bounded by a suitable fraction of the optimal social welfare will lead to the desired bound.

For any item, we set the price of item $j$ to \[ p_j = F_j^{-1} \left( 1 - q \right) \textnormal{, where } q = \frac{m \log \log n}{n}  \] and $F_j$ denotes the marginal distribution of $v_i(\{j\})$. Observe the similarity to the pricing structure in Section \ref{section:independent_valuations_static}. This allows us to prove the following theorem.

\begin{restatable}{theorem}{staticmhrmarginals}
	\label{theorem:static_mhr_marginals}
	The posted-prices mechanism with subadditive buyers and static prices is \newline $1 - O \left( \frac{\log\log\log n}{\log n} + \frac{\log m}{\log n}\right)$-competitive with respect to social welfare. 
\end{restatable}

\begin{proof}
	Lower bounding the expected revenue by a suitable fraction of the optimal social welfare will allow to prove the theorem.
	
	Recall that by our assumption $v_i(\{j\})$ is an MHR random variable for each $j \in [m]$. Therefore, we can apply Lemma~\ref{lemma:quantiles1} to get a bound on $p_j$.
	\begin{align*}
	p_j & = F_j^{-1} \left(1-q\right) \geq \frac{- \log q}{H_n} \Ex[]{\max_{i \in [n]} v_i(\{j\})} = \frac{\log n - \log \log \log n - \log m}{H_n} \Ex[]{\max_{i \in [n]} v_i(\{j\})} \\ & \geq \left( 1 - \frac{\log \log \log n + \log m + 1}{\log n} \right) \Ex[]{\max_{i \in [n]} v_i(\{j\})} .
	\end{align*} \newline
	
	Now, we aim for a lower bound on the probability that all items are sold in our mechanism. To this end, let $M^{(i)}$ denote the (random) set of items that are still unsold as buyer $i$ arrives. Observe that buyer $i$ will buy at least one item if $v_i(\{j\}) > p_j$ for some $j \in M^{(i)}$. We defined the prices such that $\Pr{v_i(\{j\}) > p_j} = q$. Consequently, $\Pr{\text{buyer $i$ buys an item} \growingmid M^{(i)} \neq \emptyset} \geq q$ for all $i$.
	
	To bound the probability of selling all items, consider the following thought experiment: For every buyer $i$, we toss a coin which shows head with probability $q$. Denote by $Z$ the random variable counting the number of occurring heads in $n$ coin tosses. By the above considerations, the probability for tossing head in our thought experiment is a lower bound on the probability that buyer $i$ buys at least one item as long as there are items remaining. As a consequence, the probability for the event of seeing at least $m$ times head is a lower bound on the probability of selling all items in our mechanism.
	
	Using that \[ \Ex[]{Z} = nq = m \log \log n \ ,\] a Chernoff bound with $\delta = 1 - \frac{1}{\log \log n}$ yields
	\begin{align*}
	\Pr{Z < m} &= \Pr{Z < (1 - \delta) \Ex{Z}} \leq \exp\left( - \frac{1}{2} \delta^2 \Ex{Z} \right) \\[0.5em]
	&= \exp\left( - \frac{1}{2} \frac{(\log \log n - 1)^2}{\log \log n} m \right) \stackrel{(1)}{\leq} \exp\left( - \log \log n + 2 \right) = \frac{\e^2}{\log n} \,,
	\end{align*}
	where in (1) we assumed that $m \geq 2$. Observe that the case of $m = 1$ is covered by our results in Section \ref{section:alpha_correlation_static}.
	%Observe, that every subadditive buyer will buy at least as many items as a unit-demand buyer. In the latter, buyer $i$ considers buying a remaining item if $v_i(\{j\})$ exceeds the price. Due to our pricing strategy, this occurs with probability $q$. Therefore,  \newline
	
	Combining all these, we can lower-bound the expected social welfare of the posted prices mechanism by 
	\begin{align*}
	\Ex[]{\ALG} & \geq \Ex[]{\revpp} \geq \Pr{\textnormal{all items are sold}} \left(\sum_{j=1}^{m} p_j\right) \geq \Pr{Z \geq m} \sum_{j=1}^{m} F_j^{-1} \left( 1- q \right) \\ & \geq \left( 1 - \frac{\e^2}{\log n} \right) \left( 1 - \frac{\log \log \log n + \log m + 1}{\log n} \right) \sum_{j=1}^{m} \Ex[]{\max_{i \in [n]} v_i(\{j\})} \\  & \geq \left( 1 - O \left( \frac{\log\log\log n}{\log n} + \frac{\log m}{\log n}\right) \right) \Ex[]{\OPT} \ .
	\end{align*}	
	
\end{proof}

Observe that the proof of Theorem \ref{theorem:static_mhr_marginals} only requires to bound the expected revenue of our mechanism. Bounding the revenue in the optimal allocation by the expected optimal social welfare, we can state the following corollary.

\begin{corollary}
	The expected revenue achieved by the posted-prices mechanism with subadditive buyers and static prices yields a $1 - O \left( \frac{\log\log\log n}{\log n} + \frac{\log m}{\log n}\right)$-fraction of the expected optimal revenue. 
\end{corollary}

\end{document}